\documentclass[11pt]{article}
\usepackage[margin=1in]{geometry}

\usepackage{amsmath, amssymb, amsthm}
\usepackage[mathic]{mathtools}

\usepackage[T1]{fontenc}
\usepackage{libertine}
\usepackage[libertine]{newtxmath} % must load after ams packages
\usepackage[scaled=0.96]{zi4} % use Inconsolata as the typewriter font

\usepackage{a4wide, dsfont, microtype, xcolor, paralist}

\usepackage[ocgcolorlinks]{hyperref} % Option ocgcolorlinks makes links only colored on screen and not on printouts
\colorlet{DarkRed}{red!50!black}
\colorlet{DarkGreen}{green!50!black}
\colorlet{DarkBlue}{blue!50!black}
\hypersetup{
	linkcolor = DarkRed,
	citecolor = DarkGreen,
	urlcolor = DarkBlue,
	bookmarksnumbered = true,
	linktocpage = true
}

\let\epsilon\varepsilon

\usepackage{xspace}

\usepackage[linesnumbered,ruled,vlined,boxed]{algorithm2e}
\newlength{\commentWidth}
\let\oldnl\nl
\newcommand{\nonl}{\renewcommand{\nl}{\let\nl\oldnl}}
\DontPrintSemicolon
\SetKwInOut{Input}{Input}\SetKwInOut{Output}{Output}

\usepackage{tikz}
\usetikzlibrary{positioning,automata,arrows.meta}
\usetikzlibrary{shapes.geometric}
\usetikzlibrary{decorations.markings}
\usetikzlibrary{positioning, shapes, arrows}

\usepackage{thmtools} 
\usepackage{thm-restate}

\usepackage{array}

\usepackage{caption}

\usepackage{subfig}

\PassOptionsToPackage{usenames,dvipsnames,svgnames}{xcolor}
\definecolor{orange}{RGB}{235,90,0}
\definecolor{darkorange}{RGB}{175,30,0}
\definecolor{turkis}{RGB}{131,182,182}
\definecolor{darkturkis}{RGB}{31,82,82}
\definecolor{green}{RGB}{102,180,0}
\definecolor{darkgreen}{RGB}{51,90,0}
\definecolor{myblue}{RGB}{0,0,213}
\definecolor{mydarkblue}{RGB}{0,0,100}
\definecolor{mybrightblue}{HTML}{74B0E4}
\definecolor{mybrighterblue}{HTML}{B3EAFA}
\definecolor{lila}{RGB}{102,0,102}
\definecolor{darkred}{RGB}{139,0,0}
\definecolor{darkyellow}{RGB}{188,135,2}
\definecolor{brightgray}{RGB}{200,200,200}
\definecolor{darkgray}{RGB}{50,50,50}
\definecolor{amaranth}{rgb}{0.9, 0.17, 0.31}
\definecolor{alizarin}{rgb}{0.82, 0.1, 0.26}
\definecolor{amber}{rgb}{1.0, 0.75, 0.0}
\definecolor{green(ryb)}{rgb}{0.4, 0.69, 0.2}
\definecolor{hanblue}{rgb}{0.27, 0.42, 0.81}
\definecolor{grannysmithapple}{rgb}{0.66, 0.89, 0.63}

\newtheorem{theorem}{Theorem}[section]
\newtheorem{lemma}[theorem]{Lemma}
\newtheorem{definition}[theorem]{Definition}
\newtheorem{corollary}[theorem]{Corollary}
\newtheorem{observation}[theorem]{Observation}

\newtheorem{problem}[theorem]{Problem}

\newenvironment{acknowledgements}{%
  \begin{abstract}
}{%
  \end{abstract}
}

\newcolumntype{C}[1]{>{\centering\let\newline\\\arraybackslash\hspace{0pt}}m{#1}}

\DeclareMathOperator{\nullval}{null}
\DeclareMathOperator{\nextval}{next}

\SetKwFunction{DFS}{DFS}
\SetKwFunction{ComputeWalks}{ComputeWalks}
\SetKwFunction{MultiSourceLongestPath}{MultiSourceLongestPath}
\SetKwFunction{MultiSourceShortestPath}{MultiSourceShortestPath}
\SetKwFunction{BFS}{BFS}
\SetKwProg{Fn}{Function}{:}{}
\SetKwRepeat{Do}{do}{while}

\title{Encoding Co-Lex Orders of Finite-State Automata in Linear Space}
\author{
 	Ruben Becker\thanks{Ca' Foscari University,
    Venice, Italy} \\ \texttt{rubensimon.becker@unive.it} \and
    Nicola Cotumaccio\thanks{University of Helsinki, Helsinki, Finland} \\ \texttt{nicola.cotumaccio@helsinki.fi} \and
    Sung-Hwan Kim\footnotemark[1] \\ \texttt{sunghwan.kim@unive.it}\and
    Nicola Prezza\footnotemark[1] \\ \texttt{nicola.prezza@unive.it} \and
    Carlo Tosoni\footnotemark[1] \\ \texttt{carlo.tosoni@unive.it}
}

\date{}

\begin{document}
\maketitle
% \begin{abstract}
%
% \end{abstract}
% \thispagestyle{empty}
% \pagebreak

\begin{abstract}
The Burrows-Wheeler transform (BWT) is a string transformation that enhances string indexing and compressibility. Cotumaccio and Prezza [SODA '21] extended this transformation to nondeterministic finite automata (NFAs) through co-lexicographic partial orders, i.e., by sorting the states of an NFA according to the co-lexicographic order of the strings reaching them. As the BWT of an NFA shares many properties with its original string variant, the transformation can be used to implement indices for locating specific patterns on the NFA itself. The efficiency of the resulting index is influenced by the width of the partial order on the states: the smaller the width, the faster the index. 
The most efficient index for \emph{arbitrary} NFAs currently known in the literature is based on the coarsest forward-stable co-lex (CFS) order of Becker et al.~[SPIRE '24].
In this paper, we prove that this CFS order can be encoded within linear space in the number of states in the automaton. 
The importance of this result stems from the fact that encoding such an order in linear space represents a big first step in the direction of building the index based on this order in near-linear time -- the biggest open research question in this context. The currently most efficient known algorithm for this task run in quadratic time in the number of transitions in the NFA and are thus infeasible to be run on very large graphs (e.g., pangenome graphs). At this point, a near-linear time algorithm is solely known for the simpler case of deterministic automata~[Becker et al., ESA '23] and, in fact, this algorithmic result was enabled by a linear space encoding for deterministic automata~[Kim et al., CPM '23]. 
\end{abstract}

\begin{acknowledgements}
\emph{Ruben Becker, Sung-Hwan Kim, Nicola Prezza, Carlo Tosoni}: funded by the European Union (ERC, REGINDEX, 101039208). Views and opinions expressed are however those of the authors only and do not necessarily reflect those of the European Union or the European Research Council Executive Agency. Neither the European Union nor the granting authority can be held responsible for them. \emph{Nicola Cotumaccio}: funded by the Helsinki Institute for Information
Technology.
\end{acknowledgements}

\newpage
\section{Introduction}

The \emph{Burrows-Wheeler transform (BWT)}~\cite{burrows1994block} is a renowned reversible string transformation that rearranges a string's characters as to improve compressibility, while at the same time allowing the implementation of efficient indices. Although the original BWT was designed for strings, Gagie et al.\ extended this transformation to a particular class of nondeterministic finite automata (NFAs), which they termed \emph{Wheeler NFAs}~\cite{gagie2017wheeler}. Subsequently, Cotumaccio et al.~\cite{cotumaccio2021indexing, cotumaccio2023jacm} managed to extend the transformation to arbitrary NFAs through the concept of \emph{co-lexicographic orders} (abbreviated to co-lex orders), yielding a natural extension of the BWT to NFAs. 
More precisely, co-lex orders are particular partial orders $\leq$ on an NFA’s states such that, if $u \leq v$, with $u, v$ being two states, then the strings reaching $u$ and not reaching $v$ are smaller than or equal to the strings reaching $v$ and not reaching $u$. Such co-lex orders exist for every NFA and can be used to implement indices on the recognized regular language. The efficiency of the index depends on the width of the used co-lex order (a parameter being equal to 1 on Wheeler NFAs and always upper-bounded by the number of states).
Specifically, the smaller is the width of the co-lex order, the faster and smaller is the resulting index. Computing the co-lex order of minimal width is however an NP-hard problem~\cite{gibney2019hardness}. This issue has been addressed by Becker et al.~\cite{becker2023sorting, tosoni2024CFS}, who introduced coarsest forward-stable co-lex (CFS) orders, a new category of partial preorders that are as useful as co-lex orders for indexing purposes. Such CFS orders are guaranteed to exist for every NFA and, furthermore, are unique and can be computed in polynomial time. Moreover, the width of the CFS order is never larger than that of any co-lex order and, in some cases, is asymptotically smaller than the minimum-width co-lex order. As a result, CFS orders enable the implementation of indices in polynomial time, which are never slower than those based on co-lex orders, and that in some cases are asymptotically faster and smaller. However, the state-of-the-art algorithm for computing such CFS orders has quadratic time complexity with respect to the number of transitions in the automaton~\cite[Corollary 1]{tosoni2024CFS}. This quadratic time complexity makes the application of such CFS orders infeasible in practice, e.g., in bioinformatics, where pangenome graphs (i.e., graphs encoding the DNA of a population) are used more and more frequently~\cite{computationalPanGenomics}. Such pangenome graphs fall within the category of big data for which only near-linear time algorithms can be considered feasible~\cite{Teng16}. For this reason, the current main open research problem in this realm is to find an efficient, i.e., near-linear time, algorithm for computing co-lex orders of small width for arbitrary NFAs. For the special case of deterministic finite automata such a near-linear time algorithm is known~\cite[Algorithm 2]{becker2023sorting} and its discovery was preceded by an encoding of this order that is in linear space with respect to the number of states of the automaton~\cite{kim2023faster}. A similar linear space representation for the general case of nondeterministic finite automata is however not known for any of the candidate co-lex orders in the literature. In this paper, we resolve this main problem that hinders us to find an efficient algorithm for computing co-lex orders for arbitrary NFAs. We do so by giving an efficient data structure for the following problem.

%\sout{More precisely, for example, as no near-linear space encoding of the CFS orders of Becker et al.~\cite{tosoni2024CFS} is known, devising an efficient algorithm to compute such orders is currently impossible. In this paper, we resolve this main problem that hinders us to find an efficient algorithm for computing co-lex orders for arbitrary NFAs by giving an efficient data structure for the following problem.}
\begin{problem}\label{main_pr}
    Given a forward-stable NFA, find a data structure for its maximum co-lex order $\leq_{FS}$ that supports queries of the form: given two states $u$ and $v$, determine if $u \leq_ {FS} v$.
\end{problem}
Here a \emph{forward-stable NFA} is an NFA for which the coarsest forward stable partition~\cite{paige1987three} is equal to the partition consisting of all singleton sets. The CFS order on an arbitrary NFA~\cite[Definition 6]{tosoni2024CFS} is actually defined as the maximum co-lex order on the corresponding forward-stable NFA. 
% This forward-stable NFA can be computed in near-linear time using the ordered partition refinement algorithm of Becker et al.~\cite{becker2023sorting} \nicola{well, the forward-stable NFA can actually already be computed with Tarjan's algorithm. The order we compute in ~\cite{becker2023sorting} is not needed to define that NFA.} that is a variant of the classical partition refinement framework of Paige and Tarjan~\cite{PaigeT87} that keeps the partition in a specific order that is consistent with the co-lex order. 
The forward stable NFA is a quotient automaton of the original automaton (thus its size is at most the size of the original automaton). As a result, a data structure for Problem~\ref{main_pr} permits to represent the CFS order of an arbitrary NFA and is thus general enough to represent a co-lex order for an arbitrary NFA.

% \vspace{-5pt}
\subparagraph{Contribution and Main Techniques.}
In this article we give a data structure for Problem~\ref{main_pr} stated above. More precisely, we prove the following theorem. In what follows, we denote with $n$ the number of states and with $m$ the number of transitions of the NFA at hand.
\begin{restatable}[]{theorem}{mainthr}\label{6:thr:final_theorem}
    Given a forward-stable NFA with $n$ states, there exists a data structure for Problem~\ref{main_pr} taking $O(n)$ space and supporting queries in $O(n)$ time.
\end{restatable}

Here space is measured in RAM words of $\Theta(\log n)$ bits. In order to better put our result in context we observe that there are two trivial solutions to Problem~\ref{main_pr}. (1)
% \begin{remark}\label{stupid_remark1}
    % Given a forward-stable NFA with $n$ states, 
    Explicitly storing the $n^2$ pairs of the co-lex order yields a data structure for Problem~\ref{main_pr} that takes $O(n^{2})$ bits and supports queries in $O(1)$ time.
% \end{remark}
% In fact, a partial order $\leq$ over a set $V$ can be represented using a binary matrix $B$, where $B[i][j] = 1$ if and only if $v_{i} \leq v_{j}$, for the $i$-th and $j$-th elements of $V$.
% \begin{remark}\label{stupid_remark2}
    (2) Storing the input NFA takes $O(m)$ space and the NFA inherently represents its maximum co-lex order $\leq_{FS}$. It is however unclear how to support queries efficiently in this case. 
% \end{remark}
% In fact, any forward-stable NFA inherently represents its maximum co-lex order $\leq$, and such an NFA can be stored using $O(m)$ space. However, it is not clear how to efficiently support queries of the form $u \leq v$ within this representation. 
% Note that both approaches of Remarks~\ref{stupid_remark1} and~\ref{stupid_remark2} 
Both of these approaches require $\Omega(n^{2})$ bits to be stored in the worst-case as the number of transitions may be $\Theta(n^{2})$. 

Our data structure that satisfies the properties in Theorem~\ref{6:thr:final_theorem} relies on three main techniques. (i)~We assume to have computed a \emph{co-lex extension} $\le$, i.e., a total order that is a superset of the maximum co-lex order of a forward-stable NFA (see Definition~\ref{linear_ext} for the formal definition).
Such a co-lex extension can be obtained by running the ordered partition refinement algorithm of Becker et al.~\cite{becker2023sorting}. (ii)~For each state, we store a \emph{left-minimal infimum walk} $P_{u}^{\inf}$ and a \emph{right-maximal supremum walk} $P_{u}^{\sup}$ to $u$. An infimum (supremum) walk to a state $u$ is a walk encoding the lexicographic smallest (largest) string reaching $u$ from the initial state. An infimum (supremum) walk $P_u$ to a state $u$ is left-minimal (right-maximal) if, whenever it intersects with another infimum (supremum) walk $P_u'$, then the predecessor in $P_u$ is smaller (larger) than or equal to the predecessor in $P_u'$ according to the co-lex extension.
We study this type of walks in Sections~\ref{section: infimum supremum walks} and~\ref{section: leftmost walk existence}. (iii)~For each state $u$, we furthermore store two integers $\phi(u,P_{u}^{\inf})$ and $\gamma(u,P_{u}^{\sup})$ that we use in order to encode the deepest states on the infimum walk $P_{u}^{\inf}$ (supremum walk $P_{u}^{\sup}$) that is in \emph{infimum} (\emph{supremum}) \emph{conflict} with $P_{u}^{\inf}$ $(P_u^{\sup})$. We introduce this concept of infimum and supremum conflicts in Section~\ref{section: inf sup conflicts}. Intuitively, a state $\bar{u}$ in $P_{u}^{\inf}$ is in infimum conflict with $P_{u}^{\inf}$ if there exists a state $\hat{u}$ that is incomparable with $\bar u$ according to the maximum co-lex order and both $\bar{u}$ and $\hat{u}$ can reach $u$ with the same string $\alpha$.

Figure~\ref{dectree} shows the decision tree used by our data structure in order to determine comparability with respect to the CFS order $\le_{FS}$ between two states $u$ and $v$ based on the above three concepts.
Let $u$ and $v$ be two states. If $u=v$, then $u\le_{FS} v$ holds by reflexivity, case (a) in Figure~\ref{dectree}. Otherwise, if $v<u$ (here $<$ denotes means $v\le u$ with respect to the co-lex extension and $v\neq u$), we can conclude $\neg (u \le_{FS} v)$, case (b) in Figure~\ref{dectree}. Otherwise, let $P_u^{\sup}$ be the right-maximal supremum walk to $u$ and let $P_{v}^{\inf}$ be the left-minimal infimum walk to $v$. Imagine traversing the two walks from $u$ and $v$ backwards yielding a sequence of pairs $(u_i, v_i)_{i\ge 1}$. While traversing the walks we can construct the strings $\sup I_u$ and $\inf I_v$. If $\sup I_u\le \inf I_v$, we know that $u\le_{FS} v$, case (c) in Figure~\ref{dectree}.
Otherwise, when comparing the pairs $u_i, v_i$ with respect to the co-lex extension $\le$, we are guaranteed to find a pair $u_j, v_j$ such that $v_j < u_j$. We distinguish two cases, if for each $i \in [j-1]$, we have $u_i < v_i$ as in case (d) of Figure~\ref{dectree}, then $\neg (u \leq_{FS} v)$. Otherwise, there exists a minimal integer $j' \leq j$ such that $u_{j'} = v_{j'}$. In this case, the comparability of $u$ and $v$ can be decided using the infimum/supremum conflicts. Consider the maximum integer $h$ such that either $u_{h}$ is in sup conflict with $P_{u}^{\sup}$ or $v_{h}$ is in inf conflict with $P_{v}^{\inf}$. If $h \geq j'$ (case (e) in Figure~\ref{dectree}), we know $\neg(u \leq_{FS} v)$. 
In fact, the opposite would imply that either $P_{u}^{\sup}$ is not right-maximal or $P_{u}^{\inf}$ is not left-minimal. 
Finally, if $h < j'$ (case (f) in Figure~\ref{dectree}), we conclude that $u \leq_{FS} v$. 
The details of the data structure are presented in Section~\ref{section: data structure}. 
We remark that the existence of our left-minimal infimum and right-maximal supremum walks is actually proved independently of the labels in the graph through an unlabeled analogue that we call leftmost/rightmost walks. These walks represent a combinatorial object in unlabeled directed graphs that may be of independent interest. As a central ingredient of our proof, we show constructively (i.e., through an algorithm) that such leftmost/rightmost walks are guaranteed to exist for any (unlabeled) directed graph and can be represented by a linear space function that encodes the predecessor of each node in their leftmost/rightmost walk. This is shown in Section~\ref{section: leftmost walk existence}. 
We proceed with preliminaries.

\begin{figure}[ht!]
\begin{center}
\resizebox{0.95\textwidth}{!}{
\begin{tikzpicture}[
dim/.style={minimum size=0.5em, align = center, font={\footnotesize}, below = 0em}, 
scale = 0.55,
s/.style={font={\footnotesize}},
dots/.style={text centered}]

    \node[dim] (0) at (-0.5,0) {$u = v$};
    \node[dim] (1) at (-0.5,-2.3) {$u \leq_{FS} v$\\Definition~\ref{def:3:colex_order}\\(a)};
    \node[dim] (2) at (3.2,0) {$v < u$};
    \node[dim] (3) at (3.2,-2.3) {$\neg(u \leq_{FS} v)$\\Corollary~\ref{lem:4:prec_pairs_colex},\\ Lemma~\ref{lem:4:inv_prop}\\(b)};
    \node[dim] (4) at (7.5,0) {$\sup I_{u} \leq \inf I_{v}$};
    
    \node[dim] (5) at (7.5,-2.3) { $u \leq_{FS} v$\\Lemma~\ref{lem:5:sup_inf_co_lex}\\(c)};

    \node[dim] (6) at (13,0) {$\forall i \in [j-1], u_{i} < v_{i}$};

    \node[dim] (7) at (21.5,0) {$\max\{ \gamma^{j'}(u, P_{u}^{\sup}), \phi^{j'}(v, P_{v}^{\inf}) \} \geq j'$};

    \node[dim] (8) at (13,-2.3) {$\neg(u \leq_{FS} v)$\\Lemma~\ref{lem:5:sup_inf_vs_walks}\\(d)};

    \node[dim] (9) at (19.5,-2.3) {$\neg(u \leq_{FS} v)$\\Lemma~\ref{lem:6:main_lemma}\\(e)};

    \node[dim] (10) at (23.5,-2.3) {$u \leq_{FS} v$\\Lemma~\ref{lem:6:main_lemma}\\(f)};

    \draw[-Stealth, s] (0) to node [left] {yes} (1);
    \draw[-Stealth, s] (0) to node [above] {no} (2);
    \draw[-Stealth, s] (2) to node [left] {yes} (3);
    \draw[-Stealth, s] (2) to node [above] {no} (4);
    \draw[-Stealth, s] (4) to node [above] {no} (6);
    \draw[-Stealth, s] (6) to node [above] {no} (7);
    \draw[-Stealth, s] (4) to node [left] {yes} (5);
    \draw[-Stealth, s] (6) to node [left] {yes} (8);
    \draw[-Stealth, s] (7) to node [left] {yes} (9);
    \draw[-Stealth, s] (7) to node [right] {no} (10);

\end{tikzpicture}
}
\end{center}
\vspace{-5mm}
\caption{Let $\leq_{FS}$ and $\leq$ be the maximum co-lex order (see Definition~\ref{def:3:colex_order}) and a co-lex extension (see Definition~\ref{linear_ext}) of a forward-stable NFA $\mathcal{A}$, respectively. Let $u$ and $v$ be any two  states in $\mathcal A$. Denote with $P_{u}^{\sup}=(u_{i})_{i\geq1}$ a supremum right-maximal walk to the state $u$ and with $P_{v}^{\inf}=(v_{i})_{i\geq1}$ an infimum left-minimal walk to $v$ (see Def.~\ref{def:3:inf_sup_walks} and~\ref{def:6:left_min_rig_max}). The figure shows the decision tree representing all possible cases that may arise when determining whether $u \leq_{FS} v$. Here, $j$ is the smallest integer such that $v_{j} < u_{j}$, while $j'$ is the smallest integer such that $u_{j'} = v_{j'}$. Functions $\phi^{j'}$ and $\gamma^{j'}$ represent the \emph{deepest} states in infimum/supremum conflict with the walks $P_u^{\sup}$ and $P_v^{\inf}$
% \vspace{-10pt}
}
\label{dectree}
\end{figure}
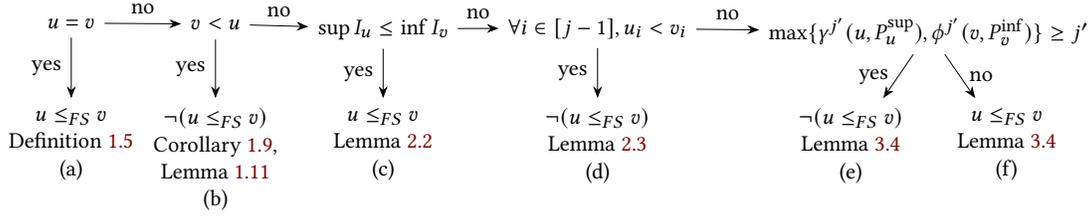

%\subparagraph{Main Techniques.}

\label{section: preliminaries}

\subparagraph{Strings and NFAs.}

Given an alphabet $\Sigma$, we denote by $\Sigma^{*}$ the set of all finite strings over $\Sigma$, where $\varepsilon \in \Sigma^{*}$ is the empty string. Moreover, we define $\Sigma^{\omega}$ as the set containing all strings formed by an infinite enumerable concatenation of characters from $\Sigma$ (i.e., strings of infinite length). In particular, we consider \textit{right-infinite} strings, meaning that $\alpha \in \Sigma^{\omega}$ is constructed from $\varepsilon$ by appending an infinite sequence of characters to its end. Therefore, the operation of prepending a character $a \in \Sigma$ of $\alpha \in \Sigma^{\omega}$ is well defined and yields the string $a\alpha$. The notation $\alpha^{\omega}$, with $\alpha \in \Sigma^{*}$, denotes the concatenation of an infinite (enumerable) number of copies of $\alpha$. In this paper, we assume to have a fixed total order $\leq$ over $\Sigma$. We extend $\leq$ to $\Sigma^{*} \cup \Sigma^{\omega}$ in order to obtain the \emph{lexicographic order} on strings.
% in the following way; given $\alpha, \beta \in \Sigma^{*} \cup \Sigma^{\omega}$, if $\alpha = \varepsilon$ or $\beta = \varepsilon$ (suppose $\alpha = \varepsilon$ w.l.o.g.) we say that $\alpha \leq \beta$. Otherwise, if both $\alpha,\beta \neq \varepsilon$, then consider $\alpha = a\alpha'$ and $\beta = b\beta'$, with $a,b \in \Sigma$. Then we say that $\alpha < \beta$ holds if $a < b$ holds, or $a = b$ and $\alpha' < \beta'$. We call this total order over $\Sigma^{*} \cup \Sigma^{\omega}$ the \textit{lexicographic order} of the strings in $\Sigma^{*} \cup \Sigma^{\omega}$. 
For each $\alpha \in \Sigma^{*} \cup \Sigma^{\omega}$, $\lvert \alpha \rvert = l$ denotes the length of $\alpha$, where $l = \infty$ if $\alpha \in \Sigma^{\omega}$. In addition, for each integer $i$ with $1 \leq i < l + 1$, $\alpha_{i}$ denotes the $i$-th character of $\alpha$, starting from the left. Finally, we denote with $\alpha[i,j]$, where $i,j$ are two integers such that $1 \leq i \leq j < l + 1$, the string of $\Sigma^{*}$ formed by the concatenation of characters $\alpha_{i} \, \alpha_{i+1} \ldots \alpha_{j}$.
% \vspace{-10pt}
%\subparagraph{NFAs and Forward-stable NFAs.}%\label{subsec:2:nfa}

A nondeterministic finite automaton (NFA) is a 4-tuple $(Q,\delta,\Sigma,s)$, where $Q$ represents the set of the states, $\delta : Q \times \Sigma \rightarrow 2^{Q}$ is the automaton's transition function, $\Sigma$ is the alphabet and $s \in Q$ is the initial state. The standard definition of NFAs also includes a set of final states; however, we omitted them since we are not interested in distinguishing between final states and non-final states. Given an NFA $\mathcal{A} = (Q, \delta, \Sigma, s)$, a state $u \in Q$, and a character $a \in \Sigma$, we may use the shortcut $\delta_{a}(u)$ for $\delta(u,a)$. We make the following assumptions on NFAs: 
$(i)$ We assume the alphabet $\Sigma$ to be effective; each character of $\Sigma$ labels at least one edge of the transition function. 
$(ii)$ We assume that every state is reachable from the initial state. 
$(iii)$ We assume that $s$ has only one incoming edge, $s \in \delta(s,\#)$, where $\# \in \Sigma$ does not label any other edge of $\mathcal{A}$. 
$(iv)$ We do not require each state to have an outgoing edge for all possible characters of $\Sigma$. 
$(v)$ We assume that our NFAs are \textit{input-consistent}; an NFA is said to be \textit{input-consistent} if all edges reaching the same state have the same label of $\Sigma$. This assumption is not restrictive since any NFA can be transformed into an input-consistent NFA by replacing each state with at most $\lvert \Sigma \rvert$ copies of itself, without changing its regular language. 
Given an NFA $\mathcal{A} = $($Q$, $\delta$, $\Sigma$, $s$), and a state $u \in Q$, with $u \neq s$, we denote with $\lambda(u)$ the (unique) character of $\Sigma$ that labels the incoming edges of $u$, thus $\lambda(s) = \#$. 
% \begin{definition}[forward-stable NFA]\label{frwstbNFAs}
    We call an NFA \emph{forward-stable} if its coarsest forward-stable partition is formed by singleton sets,
% Finally, we define \textit{forward-stable NFAs} as follows (
see the work of Becker et al.~\cite[Definition 2]{tosoni2024CFS} for the definition of forward-stable partition.
% \end{definition}

% \vspace{-10pt}
\subparagraph{Infimum and supremum walks.}%\label{subsec:3:infsup}
Given an integer $j$, we denote by $[j]$ the set $\{1,\ldots, j\}$. Given an NFA $\mathcal{A}=(Q,\delta,\Sigma,s)$, and a state $u \in Q$, we say that a \textit{walk to} $u$, denoted as $P_{u} = (u_{i})_{i=1}^{l}$, is a non-empty sequence of $l$ states from $Q$ that satisfies the following conditions: 
$(i)$ $u_{1}=u$, and 
$(ii)$ for each $i \in [l-1]$, $u_{i} \in \delta_{a}(u_{i+1})$ where $a=\lambda(u_i)$. Moreover, we denote by $P_{u} = (u_{i})_{i\geq1}$ a walk of infinite length to $u \in Q$. We define $I_{u}$ as the set of all strings $\alpha \in \Sigma^{\omega}$, for which there exists a walk $P_{u} = (u_{i})_{i\geq1}$ such that $\alpha = \lambda(u_{1}) \lambda(u_{2}) \lambda(u_{3}) \ldots$. We now provide the definition of supremum and infimum strings of an NFA's state, which was introduced by Conte et al.~\cite{conte2023dcc} and Kim et al.~\cite[Definition 7]{kim2023faster}.

%\vspace{-3pt}
\begin{definition}[Infimum and supremum strings]\label{def:2:inf_sup_str}
    Let $\mathcal{A}$ be an NFA, and let $u$ be a state of $\mathcal{A}$. Then, the \emph{infimum string} of $u$, denoted $\inf I_{u}$, is the lexicographically smallest string in $I_{u}$. The supremum string of $u$, denoted $\sup I_{u}$, instead is the lexicographically largest string  $I_{u}$.
    %\nicola{I'm not 100\% sure if this definition is well-defined: in \cite{kim2023faster} we used limits. Note that we also considered $I_u$ to contain omega-strings, like you do. Note: that paper had a mistake in the definition of infimum/supremum. It was later fixed in \url{https://doi.org/10.4230/LIPIcs.CPM.2024.1}}\carlo{It should be well-defined, I the definition from a paper of Nicola C. Unfortunately I cannot use your definition, because you also consider strings of finite length, while we do not do that here.}
\end{definition}
%\ruben{Put a comment here that the definition of Kim et al.\ and the one here are equivalent, maybe refer to Nicola C.'s paper, where the same definition is used.}
It is easy to demonstrate that for each $u \in Q$, the strings $\inf I_{u}$ and $\sup I_{u}$ exist (for instance, it follows from Observation~8 in the work of Kim et al.~\cite{kim2023faster}, see also \cite{cotumaccio2023isaac}). We now provide the definition of infimum and supremum walks. 

\begin{definition}[infimum and supremum walks]\label{def:3:inf_sup_walks}
    Let $\mathcal{A}$ be an NFA, and let $u$ be a state in $\mathcal A$. Consider $\alpha = \inf I_{u}$ and $\beta = \sup I_{u}$ and a walk $P_{u} = (u_{i})_{i\geq1}$, then we say that:
    \begin{itemize}
        \item $P_{u}$ is an \emph{infimum walk to} $u$, denoted as $P_{u}^{\inf}$, if for each integer $i \geq 1$, $\lambda(u_{i}) = \alpha_{i}$.
        
        \item $P_{u}$ is a \emph{supremum walk to} $u$, denoted as $P_{u}^{\sup}$, if for each integer $i \geq 1$, $\lambda(u_{i}) = \beta_{i}$.
    \end{itemize}

\end{definition}
% \vspace{-3pt}
In other words, a walk to $u$ is an infimum (supremum) walk, if it is labeled with the infimum (supremum) string of $u$.

% \vspace{-5pt}
\subparagraph{Co-lex orders.}

A partial order $\leq$ on a set $U$ is a reflexive, antisymmetric and transitive relation on $U$. Given a partial order $\leq$ on a set $U$, for each $u,v \in U$, we write $u < v$ if $u \leq v$ and $u \neq v$. We now report the formal definition of co-lex order of an NFA, which was first introduced by Cotumaccio and Prezza \cite[Definition 3.1]{cotumaccio2021indexing}.
\begin{definition}[Co-lex order]\label{def:3:colex_order}
    Let $\mathcal{A}$ be an NFA. A \emph{co-lex order} of $\mathcal{A}$ is a partial order $\leq$ over $Q$ that satisfies the following two axioms:
    \begin{enumerate}
        \item For each $u,v \in Q$, if $u \leq v$, then $\lambda(u) \leq \lambda(v)$.
        \item For each pair $u \in \delta_{a}(u')$ and $v \in \delta_{a}(v')$, if $u < v$, then $u' \leq v'$.
    \end{enumerate}
\end{definition}

Note that every NFA $\mathcal{A}=(Q,\delta,\Sigma,s)$ admits a co-lex order; in fact, the partial order $\leq \coloneqq \{(u,u) : u \in Q \}$ trivially satisfies the two axioms of Definition~\ref{def:3:colex_order}. We say that a co-lex order $\leq$ of an automaton $\mathcal{A}$ is the \textit{maximum} co-lex order of $\mathcal{A}$ if $\leq$ is equal to the union of all co-lex orders of $\mathcal{A}$. Becker et al.~\cite[Lemma 3]{tosoni2024CFS} showed that every forward-stable NFA admits a maximum co-lex order. Hereafter, we will denote by $\leq_{FS}$ the maximum co-lex order of an NFA. 

\begin{observation}\label{trv_obs}
    Let $\mathcal{A}$ be a forward-stable NFA, and let $u,v$ be two states of $\mathcal{A}$ such that $\lambda(u) < \lambda(v)$. Then $u <_{FS} v$.
\end{observation}

\begin{proof}
    Consider the partial order $\leq$ defined as follows; $\leq \coloneqq \{(z,z) : z \in Q \} \cup \{(u,v)\}$. $\leq$ is a co-lex order of $\mathcal{A}$, as it does not violate the axioms of Definition~\ref{def:3:colex_order}. Therefore, since $\leq_{FS}$ is defined as the union of every co-lex order of $\mathcal{A}$, it follows that $u <_{FS} v$ holds.
\end{proof}

% \vspace{-5pt}
\subparagraph{Preceding pairs.} 
We start with the definition of preceding pairs~\cite[Definition 6]{cotumaccio2022graphs}.
\begin{definition}[Preceding pairs]\label{def:4:prec_pairs}
    Let $\mathcal{A}$ be an NFA and let $(\bar{u},\bar{v}),(u,v) \in Q \times Q$ be  pairs of distinct states. We say that $(\bar{u},\bar{v})$ \emph{precedes} $(u,v)$, denoted by $(\bar{u},\bar{v}) \rightrightarrows (u,v)$, if there exist two walks, $P_{u} = (u_{i})_{i=1}^{l}$ and $P_{v} = (v_{i})_{i=1}^{l}$, such that (i)~$u_{l} = \bar{u}$ and $v_{l} = \bar{v}$, (ii)~for each $i \in [l]$, $u_{i} \neq v_{i}$, and (iii)~for each $i \in [l-1]$, $\lambda(u_{i}) = \lambda(v_{i}) = a$, for some $a \in \Sigma$.
\end{definition}

Note that if $u,v \in Q$ are distinct states, then the pair $(u,v)$ trivially precedes itself. The following observation directly follows from Definition~\ref{def:4:prec_pairs}.

% \vspace{-3pt}
\begin{observation}\label{trivial:obs:precpair}
    The relation $(\bar{u},\bar{v}) \rightrightarrows (u,v)$ is transitive. 
\end{observation}
% \vspace{-3pt}

Preceding pairs characterize the maximum co-lex order of a forward-stable NFA as follows.

% \vspace{-3pt}
\begin{restatable}[]{corollary}{cori}\label{lem:4:prec_pairs_colex}
    Let $\mathcal{A}$ be a forward-stable NFA and let $\leq_{FS}$ be its maximum co-lex order. Then  $u <_{FS} v$ holds for two distinct states $u,v$ in of $\mathcal{A}$ if and only if, for each pair $(\bar{u},\bar{v})$ with $(\bar{u},\bar{v}) \rightrightarrows (u,v)$, it holds that $\lambda(\bar{u}) \leq \lambda(\bar{v})$.
\end{restatable}
% \vspace{-3pt}

\begin{proof}
    Let $u$ and $v$ be two distinct states. Following~\cite[Lemma 7]{cotumaccio2022graphs}, $(u,v)$ is contained in the maximum co-lex relation (see~\cite[Definition 4]{cotumaccio2022graphs}) if and only if $\lambda(\bar{u}) \leq \lambda(\bar{v})$ for every pair $(\bar u, \bar v)$ with $(\bar{u},\bar{v}) \rightrightarrows (u,v)$. The maximum co-lex order $\leq_{FS}$ of a forward-stable NFA $\mathcal{A}$ is equal to its maximum co-lex relation according to~\cite[Lemma 3]{tosoni2024CFS}.
\end{proof}

Note that for any distinct states $u,v$ of $\mathcal{A}$, the statement $\neg(u <_{FS}v)$ holds if and only if there exists $(\bar{u},\bar{v}) \rightrightarrows (u,v)$ with $\lambda(\bar{u}) > \lambda(\bar{v})$. We now define \textit{co-lex extensions}.

% \vspace{-3pt}
\begin{definition}[Co-lex extension]\label{linear_ext}
    Let $\mathcal{A}$ be a forward-stable NFA and $Q$ its set of states. Consider the maximum co-lex order $\leq_{FS}$ of $\mathcal{A}$. Then, a total order $\leq$ on $Q$ is a \emph{co-lex extension} of $\mathcal{A}$, if $\leq_{FS}\; \subseteq\; \leq$.
\end{definition}
% \vspace{-3pt}

Due to Lemma 11 of the work of Becker et al.~\cite{becker2023sorting}, it follows that the \emph{ordered partition refinement}~\cite[Algorithm 1]{becker2023sorting} represents a feasible algorithm for computing a co-lex extension.

\begin{restatable}[]{lemma}{lemi}\label{lem:4:inv_prop}
    Consider a forward-stable NFA $\mathcal{A} = (Q,\delta,\Sigma,s)$. Moreover, let $\leq_{FS}$ and $\leq$ be the maximum co-lex order and a co-lex extension of $\mathcal{A}$, respectively. Then, for every $u,v \in Q$ such that $u < v$, there exists a pair $(\bar{u},\bar{v})$ with $(\bar{u},\bar{v}) \rightrightarrows (u,v)$ such that $\lambda(\bar{u}) < \lambda(\bar{v})$.
\end{restatable}

\begin{proof}
    According to Definition~\ref{linear_ext} $u < v$ implies $\neg(v <_{FS}u)$. Corollary~\ref{lem:4:prec_pairs_colex} then yields that there exists $(\bar{u}, \bar{v})$ with $(\bar{u},\bar{v}) \rightrightarrows (u,v)$ such that $\lambda(\bar{u}) < \lambda(\bar{v})$.
\end{proof}

% \section{Relation between the maximum co-lex order and infimum/supremum walks}\label{sec:4}
\section{Infimum and Supremum Walks}
\label{section: infimum supremum walks}

In this section we prove structural properties of infimum and supremum walks in a forward-stable NFA $\mathcal A$. At the end of the section we define left-minimal infimum and right-maximal supremum walks, the special type of infimum/supremum walks that our data structure is based on.
As before, we let $\leq_{FS}$ be the maximum co-lex order of $\mathcal{A}$, and $\leq$ be a co-lex extension of $\mathcal{A}$. We further denote by $n$ the number of states of $\mathcal{A}$.
The following lemma states that infimum and supremum can be compared using their first $2n - 1$ characters (see \cite{conte2023dcc, alanko2024cpm, cotumaccio2023spire} for similar results).

% \vspace{-3pt}
\begin{restatable}[]{lemma}{lemiv}\label{2nbound}
   Let $ \mathcal{A} $ be an NFA with set of states $ Q $, where $ |Q| = n $, and let $ u, v \in Q $. If $ \sup I_u \not = \inf I_v $, then $ \sup I_u [1, 2n - 1] \not = \inf {I_v} [1, 2n - 1] $.
\end{restatable}

\begin{proof}    
    Let $ C $ be the set consisting of all strings $ \sup I_u $'s and all strings $ \inf I_u $'s, and write $ C = \{\gamma_1, \gamma_2, \dots, \gamma_{n'} \} $ where $ \gamma_1 < \gamma_2 < \dots < \gamma_{n'} $. Then, $ n' \le 2n $. Notice that for every $ u \in Q $, by the maximality of $ \sup I_u $ there exist $ c_1 \in \Sigma $ and $ u_1 \in Q $ such that $ \sup I_u = c_1 \sup I_{u_1} $, and by the minimality of $ \inf I_v $ there exist $ c_2 \in \Sigma $ and $ u_2 \in Q $ such that $ \inf I_u = c_2 \inf I_{u_2} $. This implies that for every $ 2 \le i \le n' $ there exist $ c \in \Sigma $ and $ 2 \le i' \le n' $ such that $ \gamma_i = c \gamma_{i'} $. We define the array $ \mathsf{LCP}[2, n'] $ such that $ \mathsf{LCP}[i] = \mathsf{lcp}(\gamma_{i - 1}, \gamma_i) $, where $ \mathsf{lcp}(\alpha, \beta) $ is the length of the longest common prefix between $ \alpha $ and $ \beta $. To prove the lemma it is sufficient to show that $ \mathsf{LCP}[i] \le n' - 2 $ for every $i$ with $ 2 \le i \le n' $. To prove this it is sufficient to show that if $\mathsf{LCP}[i] = d$ for some $d \geq 1$ and $i$ with $2 \leq i \leq n'$, then there exists $j$ with $ 2 \le j \le n' $ such that $ \mathsf{LCP}[j] = d - 1$, because then we obtain that $ \mathsf{LCP} $ has at least $ d + 1 $ distinct entries, so $ d + 1 \le n' - 1 $, which implies $ d \le n' - 2 $. Assume that $ \mathsf{LCP}[i] = d $, with $ d \ge 1 $.  Let $ c_1, c_2 \in \Sigma $ and $ 2 \le i', i'' \le n' $ be such that $ \gamma_i = c_1 \gamma_{i'} $ and $ \gamma_{i - 1} = c_2 \gamma_{i''} $. Since $ 1 \le d = \mathsf{LCP}[i] = \mathsf{lcp}(\gamma_{i - 1}, \gamma_i) = \mathsf{lcp}(c_1 \gamma_{i'}, c_2 \gamma_{i''}) $, we have $ c_1 = c_2 $ and from $ \gamma_{i - 1} < \gamma_i $ we obtain $ \gamma_{i''} < \gamma_{i'} $, which implies $ i'' < i' $. Then, $ \mathsf{LCP}[i] = 1 + \mathsf{lcp}(\gamma_{i'},  \gamma_{i''}) = 1 + \min_{i' + 1 \le j \le i''} \mathsf{lcp}(\gamma_{j - 1},  \gamma_{j}) = 1 + \min_{i' + 1 \le j \le i''} \mathsf{LCP}[j] $, so there exists $ i' + 1 \le j \le i'' $ such that $ d = \mathsf{LCP}[i] = 1 + \mathsf{LCP}[j] $, which implies $ \mathsf{LCP}[j] = d - 1 $.
\end{proof}
% \vspace{-3pt}

%because if $ \sup I_u = \gamma_k $ and $ \inf I_v = \gamma_{k'} $, where without loss of generality $ k < k' $, then $ \mathsf{lcp}(\sup I_u, \inf I_v) = \mathsf{lcp}(\gamma_k, \gamma_{k'}) = \min_{k + 1 \le i \le k'} \mathsf{lcp}(\gamma_{i - 1}, \gamma_i) = \min_{k + 1 \le i \le k'} \mathsf{LCP}[i] \le n' - 2 \le 2n - 2 $ and so $ \sup I_u [1, 2n - 1] \not = \inf {I_v} [1, 2n - 1] $.

% We will from now on fix a forward-stable NFA $\mathcal A$, its maximum co-lex order $\le_{FS}$, and a co-lex extension $\le$. Furthermore, we will denote by $n$ the number of states in $\mathcal{A}$.
The following lemma treats the case in which the relation between the supremum and infimum of two states already implies their respective order with respect to $\le_{FS}$.
% \vspace{-3pt}
\begin{restatable}[]{lemma}{lemii}\label{lem:5:sup_inf_co_lex}
    Let $u, v$ be two distinct states. Then, $\sup I_{u} \leq \inf I_{v}$ implies $u <_{FS}v$.
\end{restatable}
% \vspace{-3pt}

\begin{proof}
     We show the contrapositive. Assume that $\neg(u <_{FS} v)$. By Corollary~\ref{lem:4:prec_pairs_colex}, this implies the existence of a pair $(\bar{u},\bar{v})$ with $(\bar{u},\bar{v})\rightrightarrows(u,v)$ such that $\lambda(\bar{u})> \lambda(\bar{v})$. 
     %Let us now consider the walks $P_{u} = (u_{i})_{i=1}^{l}$ and $P_{v} = (v_{i})_{i=1}^{l}$ for which $(\bar{u},\bar{v})\rightrightarrows(u,v)$. Therefore, $u_1 = u, v_1 = v, u_l = \bar{u}, v_l = \bar{v}$. Furthermore, consider the strings $\alpha' = \lambda(u_{1})\lambda(u_{2})\ldots \lambda(u_{l}), \beta' = \lambda(v_{1})\lambda(v_{2})\ldots \lambda(v_{l})$. By Definition~\ref{def:4:prec_pairs}, we know that $\alpha'[1,l-1] = \beta'[1,l-1]$. From $\lambda(\bar{u})> \lambda(\bar{v})$, we conclude that $\alpha' > \beta'$.
     By Definition~\ref{def:4:prec_pairs}, this implies that there exist two walks $P_{u} = (u_{i})_{i=1}^{l}$ and $P_{v} = (v_{i})_{i=1}^{l}$ with $u_1 = u, v_1 = v, u_l = \bar{u}, v_l = \bar{v}$ such that $\alpha'= \beta'$ where $\alpha' := \lambda(u_{1})\lambda(u_{2})\ldots \lambda(u_{l-1})$ and $\beta' := \lambda(v_{1})\lambda(v_{2})\ldots \lambda(v_{l-1})$.
     Now recall that every state is reachable from the initial state $s$ and $s \in \delta(s,\#)$. Thus there exist infinite strings $\alpha'' \in I_{\bar{u}}$ and $\beta'' \in I_{\bar{v}}$. Note that $\alpha''>\beta''$ because $\lambda(\bar{u})>\lambda(\bar{v})$. Since $\alpha = \alpha'\alpha''\in I_u$ and $\beta = \beta'\beta''\in I_v$, 
     %and furthermore $\alpha' > \beta'$ and $\beta'$ is not a prefix of $\alpha'$ implies $\alpha > \beta$. 
     we obtain $\sup I_{u} \geq \alpha = \alpha'\alpha'' = \beta'\alpha'' > \beta'\beta'' = \beta \geq \inf I_{v}$, completing the proof.
\end{proof}

With the next lemma we show a case in which we can determine if $\neg(u <_{FS} v)$ holds.

% \vspace{-3pt}
\begin{restatable}[]{lemma}{lemv}\label{lem:5:sup_inf_vs_walks}
    Let $u, v$ be two states with $u < v$ and $\sup I_{u} > \inf I_{v}$. Furthermore, let $P_{u}^{\sup} = (u_{i})_{i\geq1}$ and $P_{v}^{\inf} = (v_{i})_{i\geq1}$. Then there exists an integer $j$ with $1 < j < 2n$ such that $v_{j} < u_{j}$ and  $u_{i} \leq v_{i}$ as well as $\lambda(u_{i}) = \lambda(v_{i})$ for each $i \in [j-1]$. In addition, if $u_{i} < v_{i}$ holds for all $i\in [j - 1]$, then $\neg(u <_{FS} v)$.
\end{restatable}
% \vspace{-3pt}

\begin{proof}
    Consider $\alpha = \sup I_{u}$ and $\beta = \inf I_{v}$. Let $k$ be the maximal integer such that $\alpha[1,k-1] = \beta[1,k-1]$. Since $\alpha > \beta$, it holds that $\lambda(u_{k}) > \lambda(v_{k})$. By Lemma~\ref{2nbound}, $k < 2n$. Moreover, by Observation~\ref{trv_obs}, $\lambda(u_{k}) > \lambda(v_{k})$ implies $v_{k} <_{FS} u_{k}$ which, by Definition~\ref{linear_ext}, in turn implies $v_{k} < u_{k}$. Therefore, if we define $j>1$ as the smallest integer for which $v_{j} < u_{j}$, then clearly for each $i \in [j-1]$, $u_{i} \leq v_{i}$. Moreover, since $j \leq k$, and $\alpha[1,k-1] = \beta[1,k-1]$, for each $i \in [j-1]$, it holds that $\lambda(u_{i}) = \lambda(v_{i})$. For the second part of the lemma, if $u_{i} < v_{i}$ (implying $u_{i}\ne v_{i}$) for each $i \in [j-1]$, then $(u_{j},v_{j})\rightrightarrows(u,v)$ by definition. Since $v_{j} < u_{j}$, by Lemma~\ref{lem:4:inv_prop}, there exists $(\bar{u},\bar{v})\rightrightarrows(u_{j},v_{j})$ with $\lambda(\bar{u}) > \lambda(\bar{v})$. Thus, by Observation~\ref{trivial:obs:precpair} it holds that $(\bar{u},\bar{v})\rightrightarrows(u,v)$ which by Corollary~\ref{lem:4:prec_pairs_colex} proves that $\neg(u <_{FS} v)$.
\end{proof}

We summarize what we achieved show so far. Given two distinct states $u,v$, to understand whether or not $u <_{FS} v$ holds we can proceed as follows. If according to a co-lex extension $\leq$, the statement $v < u$ holds, by Definition~\ref{linear_ext}, we know $\neg(u <_{FS} v)$. If $u < v$, then if $\sup I_{u} \leq \inf I_{v}$, by Lemma~\ref{lem:5:sup_inf_co_lex}, we know that $u <_{FS} v$. Otherwise, consider a supremum walk to $u$, $(u_{i})_{i\geq1}$, and an infimum walk to $v$, $(v_{i})_{i\geq1}$. If $\sup I_{u} > \inf I_{v}$, and there exists $j$ such that $v_{j} < u_{j}$ and  $u_{i} < v_{i}$ for each $i\in [j-1]$, then by Lemma~\ref{lem:5:sup_inf_vs_walks} $\neg(u <_{FS} v)$. See Figure~\ref{fig:6:enc_ex} for an example. The only remaining case to address is the existence of an integer $j'\in [j-1]$ such that $u_{j'} = v_{j'}$. 
To study this case, we introduce left-minimal (right-maximal) walks. 

% \vspace{-5pt}
\subparagraph{Left-minimal/right-maximal walks.}
We proceed with the definition of \emph{left-minimal infimum walks} and \emph{right-maximal supremum walks}. We first define infimum and supremum graphs.

\begin{definition}[Infimum and supremum graphs]\label{def:6:inf_sup_nfa}
     The \emph{infimum (supremum) graph} $G=(Q,E)$ of $\mathcal{A}$ is the directed unlabeled graph defined as follows: (i)~The node set of $G$ is identical to the one of $\mathcal A$. (ii)~For each $u,v \in Q$, we let $(u,v) \in E$ if and only if there exists an infimum (supremum) walk $(u_{i})_{i\geq 1}$ to a state $u' \in Q$ and an integer $j$ such that $u_{j+1} = u$ and $u_{j} = v$.
\end{definition}
We can now prove the following observation.
\begin{restatable}[]{observation}{obsii}\label{obs:6:walk_inf_sup_nfa}
    Let $\mathcal{A}$ be an NFA and let $G$ be its infimum (supremum) graph. Then, every walk of infinite length in $G$ is an infimum (supremum) walk in $\mathcal{A}$.
\end{restatable}

\begin{proof}
    Let $u$ be a state of $\mathcal{A}$, and let $P_{u} = (u_{i})_{i \geq 1}$ be an arbitrary walk of infinite length. We prove this result for infimum walks, the proof for the supremum walks is analogous. Suppose for the sake of a contradiction that $P_{u}$ is not an infimum walk. We consider the largest integer $j$ for which there exists an infimum walk $P_{u}^{\inf} = (\bar{u}_{i})_{i\geq1}$, such that, for each $i\in [j]$, $\bar{u}_{i} = u_{i}$. Thus, $\bar{u}_{j+1} \neq u_{j+1}$. By Definition~\ref{def:6:inf_sup_nfa}, there exists a state $v$, an infimum walk $P_{v}^{\inf} = (v_{i})_{i\geq1}$, and an integer $k$ for which $v_{k} = u_{j}$ and $v_{k+1} = u_{j+1}$. Consider $\alpha = \lambda(\bar{u}_{1}) \ldots \lambda(\bar{u}_{j})$, $\gamma = \lambda(\bar{u}_{j+1})\lambda(\bar{u}_{j+2})\ldots\,$, $\beta = \lambda(v_{1}) \ldots \lambda(v_{k})$, $\gamma' =  \lambda(v_{k+1})\lambda(v_{k+2})\ldots\,$, we know $\inf I_{u} = \alpha\gamma$, $\inf I_{v} = \beta\gamma'$. Moreover, consider the walks $P_{u}' = (\hat{u}_{i})_{i\geq1}$, where for each $i \in [j]$, $\hat{u}_{i} = \bar{u}_{i}$, and for each $i \geq j$, $\hat{u}_{i} = v_{i+k-j}$, and $P_{v}' = (\hat{v}_{i})_{i\geq1}$, where, for each $ i \in [k]$, $\hat{v}_{i} = v_{i}$, and, for each $i \geq k$, $\hat{v}_{i} = \bar{u}_{i+j-k}$. We know that $\alpha\gamma' \in I_{u}$ and $\beta\gamma \in I_{v}$. By Definition~\ref{def:3:inf_sup_walks}, $\alpha\gamma \leq \alpha\gamma'$ which implies $\gamma \leq \gamma'$, and $\beta\gamma' \leq \beta\gamma$ which implies $\gamma' \leq \gamma$. Thus $\gamma = \gamma'$. Consequently, $P_{u}'$ is an infimum walk. However, $\hat{u}_{j} = u_{j}$ and $\hat{u}_{j+1} = u_{j+1}$, a contradiction. 
\end{proof}

We now introduce leftmost/rightmost walks in directed (unlabeled) graphs. 

\begin{definition}[Leftmost/rightmost walk]\label{def:leftmost}
     Let $ G = (V, E) $ be a directed graph, let $\leq$ be a total order on $V$ and let $ u \in V $. A walk $P_{u}= (u_{i})_{i \ge 1}$ is a \emph{leftmost (rightmost)} walk to $u$, if for each walk $\bar{P}_u=(\bar{u}_i)_{i \geq 1} $ and integer $j > 1$, $u_j = \bar{u}_j$ implies $u_{j - 1} \leq \bar{u}_{j - 1}$ $(\bar{u}_{j - 1} \leq u_{j - 1})$.
\end{definition}
We are now ready to introduce left-minimal and right-maximal walks.

\begin{definition}[Left-minimal and right-maximal walks]\label{def:6:left_min_rig_max}
    Let $u$ be a state of $\mathcal{A}$. We say that an infimum (supremum) walk $P_{u}= (u_{i})_{i\geq 1}$ is \emph{left-minimal (right-maximal)} if $P_{u}$ is a leftmost (rightmost) walk to $u$ in the infimum (supremum) graph of $\mathcal{A}$ according to $\leq$.
\end{definition}

The proof of existences of left-minimal and right-maximal walks in directed graphs is given in Theorem~\ref{exleft} in Section~\ref{section: leftmost walk existence}. The next corollary then follows together with Observation~\ref{obs:6:walk_inf_sup_nfa}.

% \vspace{-3pt}
\begin{corollary}\label{ex:leftminimal}
    Let $Q$ be the states of $\mathcal{A}$. There exists $p : Q \rightarrow Q$, such that, for each $u \in Q$, the sequence $(p^{i}(u))_{i \geq 0}$ is a left-minimal infimum (right-maximal supremum) walk to $u$.
\end{corollary}

\section{Inf and Sup Conflicts}
\label{section: inf sup conflicts}
We still consider $\mathcal{A}$ to be a fixed forward-stable NFA and $\leq_{FS}$ and $\leq$ be the maximum co-lex order and a co-lex extension of $\mathcal{A}$, respectively. 
We now define inf/sup conflicts.
\begin{definition}[inf sup conflicts]\label{def:6:min_max_confl}
    Let $u$ be a state of $\mathcal{A}$ and let $P_{u} = (u_{i})_{i\geq1}$ be an infimum (supremum) walk. For $j > 1$, we say that $u_{j}$ is in \emph{inf (sup) conflict} with $P_{u}$, denoted as $u_{j} \sqcap P_{u}$ $(u_{j} \sqcup P_{u})$, if there exists $\bar P_{u} = (\bar{u}_{i})_{i=1}^{j}$ satisfying: (i)~For each $i$ with $1 < i \leq j$, it holds that $\bar{u}_{i} \neq u_{i}$ and $\lambda(\bar{u}_{i}) = \lambda(u_{i})$. (ii)~It holds that $\neg({u}_{j} <_{FS} \bar{u}_{j})$ $(\neg(\bar{u}_{j} <_{FS} {u}_{j}))$.
\end{definition}

At this point, we present a first result concerning inf and sup conflicts.

\begin{restatable}[]{lemma}{lemvi}\label{lem:6:prev_conflict}
    Let $u$ be a state of $\mathcal{A}$, and $P_{u}=(u_{i})_{i\geq1}$ be an infimum (supremum) walk. If for some $j > 1$, $u_{j} \sqcap P_{u}$ $(u_{j} \sqcup P_{u})$ holds, then for any integer $j'$ with $1 < j' \leq j$, it holds that $u_{j'} \sqcap P_{u}$ $(u_{j'} \sqcup P_{u})$.
\end{restatable}

\begin{proof}
    We prove the lemma for infimum walks and inf conflicts; the proof for supremum walks and sup conflicts is analogous. Consider an integer $j$ such that $u_{j} \sqcap P_{u}$ and an arbitrary integer $j'$ with $ 1 < j' <j$. Moreover, let $(\bar{u}_{i})_{i=1}^{j}$ be the walk for which state $u_{j}$ is in inf conflict with $P_{u}$. Clearly, $(\bar{u}_{i})_{i=1}^{j'}$ satisfies condition~(i) of Definition~\ref{def:6:min_max_confl} for $u_{j'} \sqcap P_{u}$. It remains to prove that $\neg(u_{j'} <_{FS} \bar{u}_{j'})$. Note that for each $i$ with $j' \leq i \leq j$, we have $\bar{u}_{i} \neq u_{i}$ and $\lambda(\bar{u}_{i}) = \lambda(u_{i})$. Hence every preceding pair of $(u_{j},\bar{u}_{j})$ is a preceding pair of $(u_{j'},\bar{u}_{j'})$ and it follows from Corollary~\ref{lem:4:prec_pairs_colex} that $\neg(u_{j} <_{FS} \bar{u}_{j})$ implies $\neg(u_{j'} <_{FS} \bar{u}_{j'})$. This proves the lemma.
\end{proof}

We introduce now two functions $\phi$ and $\gamma$ which will be at the basis of our data structure.

\begin{definition}[Functions $\phi$ and $\gamma$]\label{def:6:phi_gamma}
    Let $u$ be a state of $\mathcal{A}$ and let $P_{u}^{\inf}=(u_{i})_{i\geq1}$ and $P_{u}^{\sup}=(u_{i}')_{i\geq1}$ be an infimum walk and a supremum walk to $u$, respectively. We define two functions $\phi$ and $\gamma$ as
    \[
        \phi(u, P_{u}^{\inf}) \coloneqq \max(\{i < 2n : u_{i} \sqcap P_{u}^{\inf} \} \cup \{1\} ),\ 
        \gamma(u, P_{u}^{\sup}) \coloneqq \max(\{i < 2n : u_{i}' \sqcup P_{u}^{\sup} \} \cup \{1\} ).
    \]
    Furthermore, for every integer $j > 1$, we define two functions $\phi^j$ and $\gamma^j$ as 
    \[
        \phi^j(u,P_{u}^{\inf}) \coloneqq \max_{i \in [j-1]}\{\phi(u_{i}, P_{u_{i}}^{\inf}) + i - 1 \},\ 
        \gamma^j(u,P_{u}^{\sup}) \coloneqq \max_{i \in [j-1]}\{\gamma(u_{i}',P_{u_{i}'}^{\sup}) + i - 1\},
    \]
    % \begin{enumerate}
    %     \item $\phi(u,P_{u}^{\inf},j) \coloneqq \max\{\phi(u_{i}, P_{u_{i}}^{\inf}) + i - 1 : i < j\}$
    %     \item $\gamma(u,P_{u}^{\sup},j) \coloneqq \max\{\gamma(u_{i}',P_{u_{i}'}^{\sup}) + i - 1 : i < j\}$
    % \end{enumerate}
    where $P_{u_{i}}^{\inf} = (u_{i'})_{i'\geq i}$ and $P_{u'_{i}}^{\sup} = (u'_{i'})_{i'\geq i}$.
\end{definition}

In other words, the functions $\phi^j(u,P_{u}^{\inf})$ and $\gamma^j(u,P_{u}^{\sup})$ intuitively represent the largest integer $k$ for which either $u_{k} \sqcap P_{u_{i}}^{\inf}$ or $u_{k}' \sqcup P_{u_{i}'}^{\sup}$ holds for an integer $i$ with $i < j$.

\begin{lemma}\label{lem:6:main_lemma}
    Let $u,v$ be two states of $\mathcal{A}$ with $u < v$ and $\sup I_{u} > \inf I_{v}$ and let $P_{u}^{\sup} = (u_{i})_{i\geq 1}$ and $P_{v}^{\inf} = (v_{i})_{i\geq 1}$ be a right-maximal supremum walk and a left-minimal infimum walk according to the co-lex extension $\leq$, respectively. Furthermore, assume there exists an integer $j > 1$ such that $u_{j} = v_{j}$ and $u_{i} < v_{i}$ for each $i \in [j - 1]$. Then $\neg(u <_{FS} v)$ if and only if $\max\{ \gamma^j(u, P_{u}^{\sup}), \phi^j(v, P_{v}^{\inf}) \} \geq j$.
\end{lemma}
\begin{proof}
    $(\Rightarrow)$ By Corollary~\ref{lem:4:prec_pairs_colex}, $\neg(u <_{FS} v)$ implies the existence of a pair $(\bar{u},\bar{v})$ with $(\bar{u},\bar{v})\rightrightarrows(u,v)$ such that $\lambda(\bar{u}) > \lambda(\bar{v})$. 
    Let $(\bar{u}_{i})_{i=1}^{j'}$ and $(\bar{v}_{i})_{i=1}^{j'}$ be the walks for which $(\bar{u},\bar{v})\rightrightarrows(u,v)$ holds, respectively. 
    Lemma~\ref{lem:5:sup_inf_vs_walks} implies that $\lambda(u_{i}) = \lambda(v_{i})$ for each $i \in [j]$ and $j < 2n$. Furthermore, we have $\lambda(\bar{u}_{i}) \leq \lambda(u_{i})$ and $\lambda(v_{i}) \leq \lambda(\bar{v}_{i})$ for each $i \in [j']$ by the properties of supremum and infimum walks. By the properties of preceding pairs (see Definition~\ref{def:4:prec_pairs}), for each $i \in [j]$, it holds that $\lambda(\bar{u}_{i})= \lambda(\bar{v}_{i})$ and thus $\lambda(\bar{u}_{i}) = \lambda(u_{i})$ and $\lambda(v_{i}) = \lambda(\bar{v}_{i})$. Moreover, since $\lambda(\bar{u}) > \lambda(\bar{v})$, it follows that $j' > j$.
    Due to the walks $(\bar{u}_{i})_{i=j}^{j'}$ and $(\bar{u}_{i})_{i=j}^{j'}$, it holds that $(\bar{u},\bar{v})\rightrightarrows(\bar{u}_{j},\bar{v}_j)$, which by Corollary~\ref{lem:4:prec_pairs_colex} and $\lambda(\bar{u}) > \lambda(\bar{v})$ implies $\neg(\bar{u}_{j} <_{FS} \bar{v}_{j})$.
    Now, define $h := \max\{i\in [j-1] : \bar{u}_{i} = u_{i} \}$ and $h' := \max \{ i\in[j -1] :\bar{v}_{i} = v_{i}\}$ and consider $P_{u_{h}}^{\sup} = (u_{i})_{i\geq h}$ and $P_{v_{h'}}^{\inf} = (v_{i})_{i\geq h'}$. 
    In order to prove $\max\{ \gamma^j(u, P_{u}^{\sup}), \phi^j(v, P_{v}^{\inf}) \} \geq j$, since $j < 2n$, it is thus sufficient to prove that 
     $A:=
        u_{j} \sqcup P_{u_{h}}^{\sup} \, \vee \, 
        v_{j} \sqcap P_{v_{h'}}^{\inf}$
    holds. Due to the previous considerations, $(\bar{u}_{i})_{i=h}^{j}$ satisfies condition~(i) of Definition~\ref{def:6:min_max_confl} for $u_{j} \sqcup P^{\sup}_{u_{h}}$, while $(\bar{v}_i)_{i=h'}^{j}$ satisfies condition~(i) of Definition~\ref{def:6:min_max_confl} for $v_{j} \sqcap P_{v_{h'}}^{\inf}$. Hence, $A$ holds if and only if $\neg(\bar{u}_{j} <_{FS} u_{j}) \vee \neg(v_{j} <_{FS} \bar{v}_{j})$ or equivalently $\neg(\bar{u}_{j} <_{FS} u_{j}\wedge v_{j} <_{FS} \bar{v}_{j})$. Since $u_{j} = v_{j}$ and $\leq_{FS}$ is transitive, the conclusion is implied by $\neg(\bar{u}_{j} <_{FS} \bar{v}_{j})$, which we argued above.
    
    $(\Leftarrow)$ Consider the case where the maximum is attained by $\gamma^j(u, P_{u}^{\sup})$ (the other case is analogous) and call that value $h$. We know $h \geq j$ by hypothesis. By Corollary~\ref{lem:4:prec_pairs_colex}, we have to prove the existence of a pair $(\bar{u},\bar{v})$ with $(\bar{u},\bar{v})\rightrightarrows(u,v)$ such that $\lambda(\bar{u}) > \lambda(\bar{v})$. 
    Let $k\in[j-1]$ be such that $\gamma(u_{k},P_{u_{k}}^{\sup}) + k - 1 = h$. Due to Observation~\ref{obs:6:walk_inf_sup_nfa}, the walk $P_{u_{k}}^{\sup} = (u_{i})_{i \geq k}$ is a supremum walk to $u_k$. Since $u_{h} \sqcup P^{\sup}_{u_{k}}$ and $j \leq h$, by Lemma~\ref{lem:6:prev_conflict} it follows that $u_{j} \sqcup P^{\sup}_{u_{k}}$. Therefore, by Definition~\ref{def:6:min_max_confl} there exists a walk $(\bar{u}_{i})_{i=k}^{j}$ such that $\bar{u}_{i} \neq u_{i}$ and $\lambda(\bar{u}_{i}) = \lambda(u_{i})$ for each $i$ with $k < i \leq j$ as well as $\neg(\bar{u}_{j} <_{FS} u_{j})$. Now define $\bar{u}_{i} := u_{i}$ for $i \in [k]$ and consider the two walks $(\bar{u}_{i})_{i=1}^{j}$ and $(v_{i})_{i=1}^{j}$. Note that $\bar{u}_{k} = u_k$. From Definition~\ref{def:6:min_max_confl} and Lemma~\ref{lem:5:sup_inf_vs_walks} it follows that $\lambda(\bar{u}_i) = \lambda(v_{i})$ for each $i \in [j]$.
    Suppose now that $\bar{u}_{i'} = v_{i'}$ for some $i'$ with $k < i' < j$. We can then define a new supremum walk $(\hat{u}_{i})_{i\geq1}$ to $u$ as follows. We define (i)~$\hat{u}_{i} := \bar{u}_{i}$ for $i \in [i']$, (ii)~$\hat{u}_{i} := v_{i}$ for $i$ with $i' < i < j$, and (iii)~$\hat{u}_{i} := u_{i}$ for $i \geq j$. Since $u_{j} = \hat{u}_{j}$ and $u_{j-1} < \hat{u}_{j-1}$, we conclude that $P_{u}^{\sup}$ is not a right-maximal supremum walk to $u$, a contradiction. Hence, $\bar{u}_{i} \neq v_{i}$ for all $i\in [j]$ and consequently $(\bar{u}_{j},v_{j})\rightrightarrows(u,v)$. Note that $u_{j} = v_{j}$ holds by hypothesis. Finally, since $\neg(\bar{u}_{j} <_{FS} u_{j})$, by Corollary~\ref{lem:4:prec_pairs_colex}, there exists $(\bar{u},\bar{v})$ with $(\bar{u},\bar{v})\rightrightarrows(\bar{u}_{j}, u_{j})$ such that $\lambda(\bar{u}) > \lambda(\bar{v})$. Observation~\ref{trivial:obs:precpair} then implies $(\bar{u},\bar{v})\rightrightarrows(u,v)$.
    % and completes the proof.
\end{proof}

\section{Data Structure}
\label{section: data structure}
In this section, we present our data structure and finally prove Theorem~\ref{6:thr:final_theorem}, the main theorem of this article.
Motivated by the above lemma and our previous observations, the data structure that enjoys the properties promised in Theorem~\ref{6:thr:final_theorem} can be defined as follows.
\begin{description}
    \item[Store.] For each state $u$ of $\mathcal{A}$, we store (i)~a left-minimal infimum walk $P_{u}^{\inf}$ to $u$, (ii)~a right-maximal supremum walk $P_{u}^{\sup}$ to $u$, and (iii)~the integers $\phi(u,P_{u}^{\inf})$ and $\gamma(u,P_{u}^{\sup})$. Furthermore, we store a co-lex extension $\leq$ of $\mathcal{A}$. A possible co-lex extension of $\mathcal{A}$ can be computed in $O(m \log n)$ time using the \emph{ordered partition refinement algorithm}~\cite[Algorithm 1]{becker2023sorting}, where $n$ is the number of states of $\mathcal{A}$, and $m$ the number of transition.
    \item[Query.] The procedure of a query on two arbitrary states $u,v$ is illustrated in Figure~\ref{dectree}. If $u = v$, then $u \leq_{FS} v$ by reflexivity. If $v < u$, then $\neg(u <_{FS} v)$ by Definition~\ref{linear_ext}. If $u < v$, then four possible cases may arise. (i)~Let us consider $\sup I_{u}$ and $\inf I_{v}$, which can be reconstructed from $P_{u}^{\sup}$ and $P_{v}^{\inf}$. If $\sup I_{u} \leq \inf I_{v}$ by Lemma~\ref{lem:5:sup_inf_co_lex} $u <_{FS} v$. (ii)~Otherwise, $\sup I_{u} > \inf I_{v}$ holds. We then check, if there exists $j$ such that $v_{j} < u_{j}$ and $u_{i} < v_{i}$ for all $i\in [j-1]$. If this is the case, then $\neg(u <_{FS} v)$ according to Lemma~\ref{lem:5:sup_inf_vs_walks}. Otherwise, we let $j'$ with $1 < j' < j$ be such that $u_{j'} = v_{j'}$ and $u_{i} < v_{i}$ for all $i\in [j'-1]$. (iii)~By Lemma~\ref{lem:6:main_lemma}, if $h:=\max\{ \gamma^{j'}(u, P_{u}^{\sup}), \phi^{j'}(v, P_{v}^{\inf}) \} < j'$, we know that $u <_{FS} v$, (iv)~while if $h \geq j'$, we know that $\neg(u <_{FS} v)$. 
\end{description}
We refer the reader to Figure~\ref{fig:6:enc_ex} for an example of this data structure.
The correctness of our data structure follows immediately from the description above. In order to establish Theorem~\ref{6:thr:final_theorem}, it remains to argue why the space and query time bounds hold.
\begin{proof}[Proof of Theorem~\ref{6:thr:final_theorem}]
    By Corollary~\ref{ex:leftminimal}, we know that this information can be stored in $O(n)$ space.
    By Lemma~\ref{2nbound}, we can check $\sup I_{u} \leq \inf I_{v}$ in $O(n)$ time, and if this is the case by Lemma~\ref{lem:5:sup_inf_co_lex} $u <_{FS} v$.
    Otherwise, by Lemma~\ref{lem:5:sup_inf_vs_walks}, we can indentify in $O(n)$ time the integer $j < 2n$ such that $v_j < u_j$ and $u_i \leq v_i$ for all $i \in [j-1]$.
    If $u_i < v_{j}$ for all $i \in [j-1]$, then $\neg(u <_{FS} v)$ by Lemma~\ref{lem:5:sup_inf_vs_walks}. 
    On the other hand, if there exists $j'$ with (i)~$j' < j$, (ii)~$u_{j'} = v_{j'}$, and (iii)~$u_{i} < v_{i}$ for all $i \in [j'-1]$, by Definition~\ref{def:6:phi_gamma}, since $j' < 2n$, we can compute $\gamma^{j'}(u,P_u^{\sup})$ and $\phi^{j'}(v,P_v^{\inf})$ in $O(n)$ time.
    Hence the total time complexity is $O(n)$.
\end{proof}

\captionsetup[subfloat]{position=top,labelformat=empty}
\begin{figure}[!ht]
\centering
\resizebox{0.97\textwidth}{!}{
\subfloat[]{
\begin{tikzpicture}[
dim/.style={minimum size=2em, font={\small}}, 
scale = 0.80,
dots/.style={text centered}]

    \node[dots] (0) at (0,-2.1) {$(a)$};

    \node[state,dim] (1) at (0,0) {$1$};
    \node[state,dim] (9) at (0,1.4) {$9$};
    \node[state,dim] (13) at (-1.4,1.4) {$13$};
    \node[state,dim] (8) at (1.4,1.4) {$8$};
    \node[state,dim] (2) at (-1.4,2.8) {$2$};
    \node[state,dim] (6) at (0,2.8) {$6$};
    \node[state,dim] (12) at (1.4,2.8) {$12$};
    \node[state,dim] (5) at (2.6,2.8) {$5$};
    \node[state,dim] (3) at (-1.4,4.2) {$3$};
    \node[state,dim] (4) at (0,4.2) {$4$};
    \node[state,dim] (7) at (1.4,4.2) {$7$};
    \node[state,dim] (10) at (-1.4,5.8) {$10$};
    \node[state,dim] (11) at (2.6,5.8) {$11$};

    \draw[-Stealth] (1) to[loop below] node [below] {$\#$} (1);
    \draw[-Stealth] (1) to node [right] {$b$} (9);
    \draw[-Stealth] (1) to (1.4,0) to node [left] {$b$} (8);
    \draw[-Stealth] (1) to (-1.5,0) to[bend left=55] node [left] {$a$} (2);
    \draw[-Stealth] (2) to node [left] {$a$} (3);
    \draw[-Stealth] (3) to node [left] {$b$} (10);
    \draw[-Stealth] (4) to (0,5.8) to node [above] {$b$} (10);
    \draw[-Stealth] (5) to node [left] {$b$} (11);
    \draw[-Stealth] (6) to node [left] {$a$} (4);
    \draw[-Stealth] (6) to node [left] {$a$} (7);
    \draw[-Stealth] (7) to node [left] {$b$} (11);
    \draw[-Stealth] (8) to node [left] {$a$} (5);
    \draw[-Stealth] (8) to node [left] {$b$} (12);
    \draw[-Stealth] (9) to[bend left = 45] node [below] {$b$} (13);
    \draw[-Stealth] (9) to node [right] {$a$} (6);
    \draw[-Stealth] (12) to node [left] {$a$} (7);
    \draw[-Stealth] (13) to node [left] {$a$} (2);
    \draw[-Stealth] (13) to[bend left = 45] node [above] {$b$} (9);

    \begin{scope}[xshift=16em]

    \node[dots] (0) at (0,-2.1) {$(b)$};
    
    \node[state,dim] (1) at (0,0) {$1$};
    \node[state,dim] (9) at (0,1.4) {$9$};
    \node[state,dim] (13) at (-1.4,1.4) {$13$};
    \node[state,dim] (8) at (1.4,1.4) {$8$};
    \node[state,dim] (2) at (-1.4,2.8) {$2$};
    \node[state,dim] (6) at (0,2.8) {$6$};
    \node[state,dim] (12) at (1.4,2.8) {$12$};
    \node[state,dim] (5) at (2.6,2.8) {$5$};
    \node[state,dim] (3) at (-1.4,4.2) {$3$};
    \node[state,dim] (4) at (0,4.2) {$4$};
    \node[state,dim] (7) at (1.4,4.2) {$7$};
    \node[state,dim] (10) at (-1.4,5.8) {$10$};
    \node[state,dim] (11) at (2.6,5.8) {$11$};

    \draw[-Stealth] (1) to[loop below] node [below] {$\#$} (1);
    \draw[-Stealth] (1) to node [right] {$b$} (9);
    \draw[-Stealth] (1) to (1.4,0) to node [left] {$b$} (8);
    \draw[-Stealth] (1) to (-1.5,0) to[bend left=55] node [left] {$a$} (2);
    \draw[-Stealth] (2) to node [left] {$a$} (3);
    \draw[-Stealth] (3) to node [left] {$b$} (10);
    %\draw[-Stealth] (4) to (0,5.8) to node [above] {$b$} (10);
    %\draw[-Stealth] (5) to node [left] {$b$} (11);
    \draw[-Stealth] (6) to node [left] {$a$} (4);
    \draw[-Stealth] (6) to node [left] {$a$} (7);
    \draw[-Stealth] (7) to node [left] {$b$} (11);
    \draw[-Stealth] (8) to node [left] {$a$} (5);
    \draw[-Stealth] (8) to node [left] {$b$} (12);
    \draw[-Stealth] (9) to[bend left = 45] node [below] {$b$} (13);
    \draw[-Stealth] (9) to node [right] {$a$} (6);
    %\draw[-Stealth] (12) to node [left] {$a$} (7);
    %\draw[-Stealth] (13) to node [left] {$a$} (2);
    %\draw[-Stealth] (13) to[bend left = 45] node [above] {$b$} (9);
    \end{scope}

    \begin{scope}[xshift=32em]

    \node[dots] (0) at (0,-2.1) {$(c)$};
    
    \node[state,dim] (1) at (0,0) {$1$};
    \node[state,dim] (9) at (0,1.4) {$9$};
    \node[state,dim] (13) at (-1.4,1.4) {$13$};
    \node[state,dim] (8) at (1.4,1.4) {$8$};
    \node[state,dim] (2) at (-1.4,2.8) {$2$};
    \node[state,dim] (6) at (0,2.8) {$6$};
    \node[state,dim] (12) at (1.4,2.8) {$12$};
    \node[state,dim] (5) at (2.6,2.8) {$5$};
    \node[state,dim] (3) at (-1.4,4.2) {$3$};
    \node[state,dim] (4) at (0,4.2) {$4$};
    \node[state,dim] (7) at (1.4,4.2) {$7$};
    \node[state,dim] (10) at (-1.4,5.8) {$10$};
    \node[state,dim] (11) at (2.6,5.8) {$11$};

    \draw[-Stealth] (1) to[loop below] node [below] {$\#$} (1);
    %\draw[-Stealth] (1) to node [right] {$b$} (9);
    \draw[-Stealth] (1) to (1.4,0) to node [left] {$b$} (8);
    %\draw[-Stealth] (1) to (-1.5,0) to[bend left=55] node [left] {$a$} (2);
    \draw[-Stealth] (2) to node [left] {$a$} (3);
    %\draw[-Stealth] (3) to node [left] {$b$} (10);
    \draw[-Stealth] (4) to (0,5.8) to node [above] {$b$} (10);
    %\draw[-Stealth] (5) to node [left] {$b$} (11);
    \draw[-Stealth] (6) to node [left] {$a$} (4);
    %\draw[-Stealth] (6) to node [left] {$a$} (7);
    \draw[-Stealth] (7) to node [left] {$b$} (11);
    \draw[-Stealth] (8) to node [left] {$a$} (5);
    \draw[-Stealth] (8) to node [left] {$b$} (12);
    \draw[-Stealth] (9) to[bend left = 45] node [below] {$b$} (13);
    \draw[-Stealth] (9) to node [right] {$a$} (6);
    \draw[-Stealth] (12) to node [left] {$a$} (7);
    \draw[-Stealth] (13) to node [left] {$a$} (2);
    \draw[-Stealth] (13) to[bend left = 45] node [above] {$b$} (9);
    \end{scope}
\end{tikzpicture}
}\hspace{1em}\subfloat[]{
\renewcommand{\arraystretch}{1.1}
\begin{tabular}{| C{1em} C{4.2em} C{4.2em} C{1em} C{1em} |}
    %\hline
    %\multicolumn{5}{|c|}{\textbf{Encoding}} \\
    \hline
    $u$ & $\inf I_{u}$ & $\sup I_{u}$ & $\phi$ & $\gamma$ \\
    \hline
    $1$ & $\#^\omega$ & $\#^\omega$ & 1 & 1 \\ 
    $2$ & $a\#^\omega$ & $ab^{\omega}$ & 1 & 1 \\ 
    $3$ & $aa\#^\omega$ & $aab^{\omega}$ & 1 & 1 \\
    $4$ & $aab\#^\omega$ & $aab^{\omega}$ & 1 & 1 \\
    $5$ & $ab\#^{\omega}$ & $ab\#^{\omega}$ & 1 & 1 \\
    $6$ & $ab\#^\omega$ & $ab^{\omega}$ & 1 & 1 \\
    $7$ & $aab\#^\omega$ & $abb\#^{\omega}$ & 1 & 1 \\
    $8$ & $b\#^{\omega}$ & $b\#^{\omega}$ & 1 & 1 \\
    $9$ & $b\#^{\omega}$ & $b^{\omega}$ & 1 & 1 \\
    $10$ & $baa\#^{\omega}$ & $baab^{\omega}$ & 1 & 25 \\
    $11$ & $baab\#^{\omega}$ & $babb\#^{\omega}$ & 2 & 1\\
    $12$ & $bb\#^{\omega}$ & $bb\#^{\omega}$ & 1 & 1\\
    $13$ & $bb\#^{\omega}$ & $b^{\omega}$ & 1 & 1 \\
    \hline
\end{tabular}
}}
\caption{Consider the forward-stable NFA $\mathcal{A}$ in Figure (a). Each state is assigned an integer $i$ indicating its position in the co-lex extension $\leq$. We denote by $u_{i}$ the $i$-th state according to $\leq$. Figures (b) and (c) show the NFAs encoding a left-minimal infimum walk and a right-maximal supremum walk, respectively, for each state. The table on the right shows for each state $u$ the values of $\inf I_{u}$, $\sup I_{u}$, $\phi(u,P_{u}^{\inf})$, and $\gamma(u,P_{u}^{\sup})$, where $P_{u}^{\inf}$ and $P_{u}^{\sup}$ are the walks shown in Figures (b) and (c). Our data structure comprises $\le$, the walks in Figures (b) and (c), and the two columns $\phi$, $\gamma$ from the table. We sketch the four cases that arise when determining whether $u <_{FS} v$ holds, assuming $u < v$. (i)~By Lemma~\ref{lem:5:sup_inf_co_lex}, since $\sup I_{u_3} \leq \inf I_{u_5}$, it follows that $u_3 <_{FS} u_{5}$. (ii)~Consider $P_{u_{2}}^{\sup} = u_2,u_{13}\ldots$ and $P_{u_6}^{\inf} = u_6,u_9\ldots$, since $\sup I_{u_2} > \inf I_{u_6}$, and $u_2 < u_6, u_{13} > u_9$, by Lemma~\ref{lem:5:sup_inf_vs_walks}, $\neg(u_{2} <_{FS} u_{6})$. (iii)~Consider now $P_{u_4}^{\sup} = u_{4},u_{6}\ldots$ and $P_{u_7}^{\inf} = u_{7},u_{6}\ldots$. Since, $\sup I_{u_4} > \inf I_{u_7}$, $u_{4} < u_{7},u_{6} = u_6$, and $\max\{ \gamma^2(u_4, P_{u_4}^{\sup}), \phi^2(u_7, P_{u_7}^{\inf}) \} = 1 < 2$, by Lemma~\ref{lem:6:main_lemma}, we can conclude $u_4 <_{FS} u_7$. (iv)~Finally, consider $P_{u_{10}}^{\sup} = u_{10},u_{4},u_{6}\ldots$ and $P_{11}^{\inf} = u_{11},u_{7},u_{6}\ldots$, due to the fact that $\sup I_{u_{10}} > \inf I_{u_{11}}$, $u_{10} < u_{11}$, $u_{4} < u_{7}$, $u_{6} = u_{6}$, and $\max\{\gamma^3(u_{10}, P_{u_{10}}^{\sup}), \phi^3(u_{11}, P_{u_{11}}^{\inf}) \} = 26 \geq 3$, by Lemma~\ref{lem:6:main_lemma}, we conclude that $\neg(u_{10} <_{FS} u_{11})$.
% \vspace{-10pt}
}
\label{fig:6:enc_ex}
\end{figure}
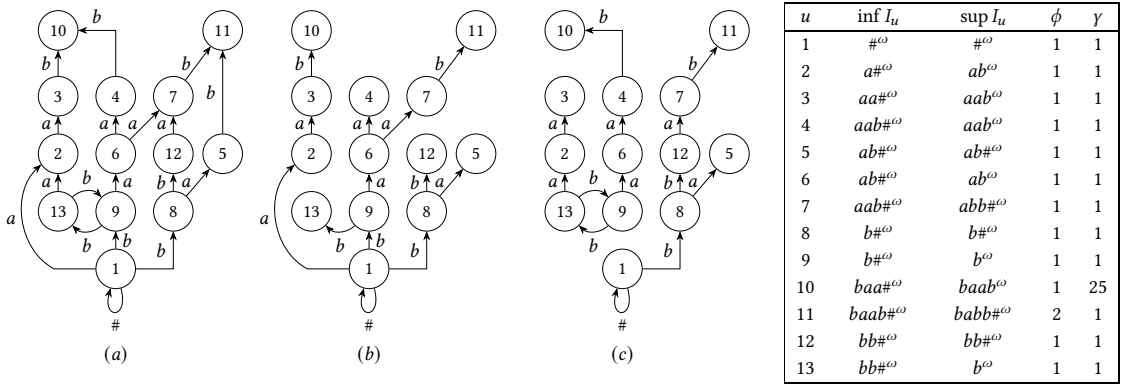

\section{Existence of a Leftmost Walk}
\label{section: leftmost walk existence}

In this section, we consider an (unlabeled) directed graph $ G = (V, E) $, and a total order $ \le $ on $ V $. We always assume that every node of $ G $ has at least one incoming edge. Moreover, a walk to a node $u \in V$, denoted by $P_u = (u_i)_{i=1}^{l}$, is a sequence of $l$ nodes such that: (i)~$u_{1} = u$, and (ii)~$(u_{i+1},u_{i}) \in E$ for each $i \in [l-1]$. We denote by $P_u = (u_i)_{i \geq 1}$ a walk to $u$ of infinite length. We fix some further graph-related notation. The \emph{induced subgraph of $G$ on $V' \subseteq V$} is the graph $G[V'] := (V', E\cap(V'\times V'))$. The \emph{subgraph of $G$ reachable from a subset $S\subseteq V$}, denoted by $\delta_G(S)$, is the induced subgraph $G[V']$ on the nodes $V'=\{v\in V:\exists P_v = (v_i)_{i = 1}^\ell \text{ with } v_\ell\in S\}$. For two sets of nodes $S,T\subseteq V$, we call $N(S):=\{v\in V: \exists u\in S \text{ and }(u,v)\in E\}$ the neighbors of $S$, $E(S):=E\cap (S\times V)$ the edges from $S$ to $N(S)$, and $E(S,T):=E\cap (S\times T)$ the edges from $S$ to $T$. For two directed graphs $G_1=(V_1, E_1)$ and $G_2=(V_2, E_2)$ their union is defined as $G_1 \cup G_2 := (V_1 \cup V_2, E_1 \cup E_2)$.

The rest of the section is devoted to proving the following theorem. 
\begin{theorem}\label{exleft}
    Let $ G = (V, E) $ be a directed graph, and let $\le $ be a total order on $ V $. Then, there exists a function $ p: V \rightarrow V $ such that $P_{u} = (p^i (u))_{i\ge 0} $ is a leftmost (rightmost) walk to $u$ for every $ u \in V $.
\end{theorem}

Hereafter, we only consider the leftmost case since the rightmost is analogous. We prove Theorem~\ref{exleft} constructively, i.e., we give an algorithm, termed \emph{Forward Visit}, that computes the function $p$ on input a directed graph $G = (V,E)$ and a total order $\leq$ on $V$. At a very high level, the algorithm consists of DFS and BFS visits that alternate with each other. Here the DFS visit has the purpose of computing a cycle $C$. Once this cycle $C$ is computed, we start two BFS-like searches starting from $C$ on two different subgraphs $G_L$ and $G_R$ of $G$. These two BFS-like searches each compute walks starting from nodes in $C$ that when concatenated with $C$ form the leftmost infinite walks for all nodes on those walks. These walks will then be represented by the function $p$. More precisely, the BFS-like searches work as follows: Let $V'$ be nodes for which we have not computed a value for $p$ yet. The graph $G_L$ contains the nodes that can be reached in $G[V']$ from some node $u\in C$ through a neighbor node $v\notin C$ that is on the \emph{left} of $u$'s cycle successor with respect to the total order $\le$. The subgraph $G_R$ is defined symmetrically as reachable through \emph{right} neighbors. The BFS-like searches are then a \emph{multi-source shortest path} search in $G_R$ and a \emph{multi-source longest path} search in $G_L$. The intuition why we compute shortest paths for $G_R$ and longest paths for $G_L$ is as follows. Let $z\notin C$ be a node in $G_L$ that is reachable by a length-$d$ path from some cycle node $u\in C$ through a \emph{left neighbor} $v\notin C$ of $u$. Now, consider some other length-$d$ path from $u$ to $z$ that follows the cycle longer, i.e., goes through the successor of $u$ in the cycle. Such a walk is not left-most by definition and in fact our algorithm will never output such a walk as it is constructed using a strictly shorter path from $C$ to $z$ than the one going through $v$. By a symmetric argument nodes in $G_R$ should be connected to $C$ via shortest path.

% and $E_{\le}(S,T):=\{(u,v)\in E_(S,T): u\le v\}$.
% provide the following definitions for a fixed directed graph $G = (V,E)$.
% \begin{definition}[Induced subgraph]
%     For any $V' \subseteq V$, the induced subgraph of $G$ on $V'$, denoted by $G[V'] = (V', E')$, is the directed graph such that for each $(u,v) \in V' \times V'$, $(u,v) \in E'$ if and only if $(u,v) \in E$.
% \end{definition}

% \begin{definition}[Reachable subgraph]
%     For any $V' \subseteq V$, the reachable subgraph of $G$ from $V'$, denoted by $\delta_{G}(V') = (\hat{V},\hat{E})$, is the directed graph such that: (i)~For each $u \in V$, $u \in \hat{V}$ if and only if there exists a walk in $G$, $P_{u} = (u_{i})_{i=1}^{l}$, with $u_{l} \in V'$. (ii)~For each $(u,v) \in V \times V$, $(u,v) \in \hat{E}$ if and only if there exists a walk in $G$, $(u_{i})_{i=1}^{l}$, with $u_{1} = v$, $u_{2} = u$, and $u_{l} \in V'$.
% \end{definition}

% The image was here.
\begin{figure}[ht!]
\begin{center}
\resizebox{0.85\textwidth}{!}{
\begin{tikzpicture}[
dim/.style={minimum size=1.85em, font={\footnotesize}}, 
scale = 0.55,
cycle/.style={color=darkgreen},
r/.style={color=blue},
l/.style={color=darkred},
dots/.style={text centered}]
    \node[dots] (0) at (-1.5,-2.7) {$(a)$};

    \node[state, dim] (6) at (0,0) {$6$};
    \node[state, dim] (2) at (0,-1.8) {$2$};
    \node[state, dim] (11) at (0,-3.6) {$11$};
    \node[state, dim] (1) at (0,-5.4) {$1$};
    \node[state, dim] (3) at (1.7,0) {$3$};
    \node[state, dim] (5) at (1.7,-1.8) {$5$};
    \node[state, dim] (9) at (1.7,-3.6) {$9$};
    \node[state, dim] (16) at (1.7,-5.4) {$16$};
    \node[state, dim] (4) at (3.4,0) {$4$};
    \node[state, dim] (7) at (3.4,-1.8) {$7$};
    \node[state, dim] (13) at (3.4,-3.6) {$13$};
    \node[state, dim] (12) at (3.4,-5.4) {$12$};
    \node[state, dim] (10) at (5.1,0) {$10$};
    \node[state, dim] (8) at (5.1,-1.8) {$8$};
    \node[state, dim] (15) at (5.1,-3.6) {$15$};
    \node[state, dim] (14) at (5.1,-5.4) {$14$};

    \draw[-Stealth] (2) to [bend left = 30] (6);
    \draw[-Stealth] (6) to [bend left = 30] (2);
    \draw[-Stealth] (2) to (11);
    \draw[-Stealth] (2) to (5);
    \draw[-Stealth] (11) to (9);
    \draw[-Stealth] (5) to (11);
    \draw[-Stealth] (5) to (3);
    \draw[-Stealth] (9) to (5);
    \draw[-Stealth] (3) to (4);
    \draw[-Stealth] (5) to [bend left = 30] (7);
    \draw[-Stealth] (7) to [bend left = 30] (5);
    \draw[-Stealth] (9) to (13);
    \draw[-Stealth] (7) to (4);
    \draw[-Stealth] (7) to (13);
    \draw[-Stealth] (4) to (15);
    \draw[-Stealth] (4) to (10);
    \draw[-Stealth] (13) to (15);
    \draw[-Stealth] (15) to (8);
    \draw[-Stealth] (10) to (8);
    \draw[-Stealth] (1) to (9);
    \draw[-Stealth] (12) to (16);
    \draw[-Stealth] (16) to (1);
    \draw[-Stealth] (16) to (9);
    \draw[-Stealth] (12) to [bend left = 30] (14);
    \draw[-Stealth] (14) to [bend left = 30] (12);

    \begin{scope}[xshift=28em]

    \node[dots] (0) at (-1.5,-2.7) {$(b)$};

    \node[state, dim] (6) at (0,0) {$6$};
    \node[state, dim] (2) at (0,-1.8) {$2$};
    \node[state, dim, r] (11) at (0,-3.6) {$11$};
    \node[state, dim] (1) at (0,-5.4) {$1$};
    \node[state, dim, l] (3) at (1.7,0) {$3$};
    \node[state, dim, cycle] (5) at (1.7,-1.8) {$5$};
    \node[state, dim, r] (9) at (1.7,-3.6) {$9$};
    \node[state, dim] (16) at (1.7,-5.4) {$16$};
    \node[state, dim, l] (4) at (3.4,0) {$4$};
    \node[state, dim, cycle] (7) at (3.4,-1.8) {$7$};
    \node[state, dim, r] (13) at (3.4,-3.6) {$13$};
    \node[state, dim] (12) at (3.4,-5.4) {$12$};
    \node[state, dim, l] (10) at (5.1,0) {$10$};
    \node[state, dim, l] (8) at (5.1,-1.8) {$8$};
    \node[state, dim, l] (15) at (5.1,-3.6) {$15$};
    \node[state, dim] (14) at (5.1,-5.4) {$14$};

    \draw[-Stealth] (2) to [bend left = 30] (6);
    \draw[-Stealth] (6) to [bend left = 30] (2);
    \draw[-Stealth] (2) to (11);
    \draw[-Stealth] (2) to (5);
    \draw[-Stealth, r] (11) to (9);
    \draw[-Stealth, r] (5) to (11);
    \draw[-Stealth, l] (5) to (3);
    \draw[-Stealth, r] (9) to (5);
    \draw[-Stealth, l] (3) to (4);
    \draw[-Stealth, cycle] (5) to [bend left = 30] (7);
    \draw[-Stealth, cycle] (7) to [bend left = 30] (5);
    \draw[-Stealth, r] (9) to (13);
    \draw[-Stealth, l] (7) to (4);
    \draw[-Stealth, r] (7) to (13);
    \draw[-Stealth, l] (4) to (15);
    \draw[-Stealth, l] (4) to (10);
    \draw[-Stealth] (13) to (15);
    \draw[-Stealth, l] (15) to (8);
    \draw[-Stealth, l] (10) to (8);
    \draw[-Stealth] (1) to (9);
    \draw[-Stealth] (12) to (16);
    \draw[-Stealth] (16) to (1);
    \draw[-Stealth] (16) to (9);
    \draw[-Stealth] (12) to [bend left = 30] (14);
    \draw[-Stealth] (14) to [bend left = 30] (12);
        
    \end{scope}

    \begin{scope}[xshift=56em]

    \node[dots] (0) at (-1.5,-2.7) {$(c)$};

    \node[state, dim] (6) at (0,0) {$6$};
    \node[state, dim] (2) at (0,-1.8) {$2$};
    \node[state, dim] (11) at (0,-3.6) {$11$};
    \node[state, dim] (1) at (0,-5.4) {$1$};
    \node[state, dim] (3) at (1.7,0) {$3$};
    \node[state, dim] (5) at (1.7,-1.8) {$5$};
    \node[state, dim] (9) at (1.7,-3.6) {$9$};
    \node[state, dim] (16) at (1.7,-5.4) {$16$};
    \node[state, dim] (4) at (3.4,0) {$4$};
    \node[state, dim] (7) at (3.4,-1.8) {$7$};
    \node[state, dim] (13) at (3.4,-3.6) {$13$};
    \node[state, dim] (12) at (3.4,-5.4) {$12$};
    \node[state, dim] (10) at (5.1,0) {$10$};
    \node[state, dim] (8) at (5.1,-1.8) {$8$};
    \node[state, dim] (15) at (5.1,-3.6) {$15$};
    \node[state, dim] (14) at (5.1,-5.4) {$14$};

    \draw[-Stealth] (2) to [bend left = 30] (6);
    \draw[-Stealth] (6) to [bend left = 30] (2);
    \draw[-Stealth] (11) to (9);
    \draw[-Stealth] (5) to (11);
    \draw[-Stealth] (5) to (3);
    \draw[-Stealth] (3) to (4);
    \draw[-Stealth] (5) to [bend left = 30] (7);
    \draw[-Stealth] (7) to [bend left = 30] (5);
    \draw[-Stealth] (7) to (13);
    \draw[-Stealth] (4) to (15);
    \draw[-Stealth] (4) to (10);
    \draw[-Stealth] (10) to (8);
    \draw[-Stealth] (12) to (16);
    \draw[-Stealth] (16) to (1);
    \draw[-Stealth] (12) to [bend left = 30] (14);
    \draw[-Stealth] (14) to [bend left = 30] (12);
    \end{scope}
\end{tikzpicture}
}
\end{center}
\vspace{-3mm}
\caption{(a) A directed graph $G = (V,E)$ and a total order $\leq$ over $V$ represented by the integer names of nodes. (b) In green the first cycle $C$ that is found by the DFS of Algorithm~\ref{alg:6:left}, if we start a DFS from node $2$; in red and blue the subgraphs $G_{L}$ and $G_{R}$ corresponding to $C$, respectively. Here, $L = \{3, 4\}$ and $R = \{11, 13\}$. (c) Leftmost walks represented by $p$ (indicated by the shown edges).
% \vspace{-15pt}
}
\label{fig:6:left_ex}
\end{figure}
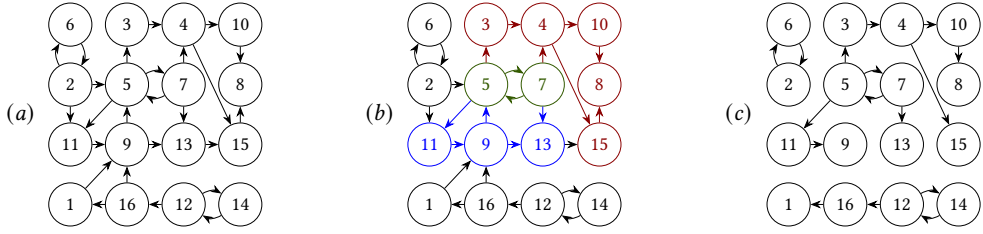

% \vspace{-5pt}
\subparagraph{The Forward Visit Algorithm.}
We proceed with a detailed description of the algorithm, a pseudocode implementation can be found in Algorithm~\ref{alg:6:left}. 
For each $u \in V$, the algorithm maintains three values, $p(u)$ (initially $\nullval$), $u.color$ (initially white), and $u.\nextval$ (initially $\nullval$). The algorithm iterates over the nodes in $V$ in an arbitrary order. For each $u \in V$ with $u.color$ = white, the algorithm starts a DFS visit from $u$. This DFS first changes the color of $u$ to gray, and eventually to black once the DFS from $u$ terminated. The DFS visit from $u$ is a classical DFS up to two caveats. (1)~For a node $u\in V$, its neighbors $N(u)$ are processed in increasing order with respect to $\leq$. (2)~The color $v.color$ of a neighbor $v\in N(u)$ determines our action when processing $v$ as follows: (a)~if $v.color =$ black, the node $v$ is ignored, (b)~if $v.color = $ white, we process the node normally, i.e., we launch a new DFS from $v$ and set $u.\nextval = v$, and (c)~if $v.color =$ gray, we still set $u.\nextval = v$ but we temporarily interrupt the DFS from $u$. Case~(c) implies that the DFS from $u$ found a cycle $C$ in $G$ that can be reconstructed from the $\nextval$-values. We now call the function \texttt{ComputeWalks} with input $C$.

\SetKwRepeat{Do}{do}{while}
\begin{algorithm}[!ht]
\LinesNumbered
\DontPrintSemicolon
\caption{Forward Visit}\label{alg:6:left}
    \Input{directed graph $G=(V,E)$, total order $\leq$ over $V$}
    \Output{function $p : V \rightarrow V$ s.t.\ $(p^i (u))_{i\ge 0} $ is a leftmost walk for every $ u \in V $}

    \smallskip

    \Fn{\ComputeWalks{$C$}}{
    
        $V' \gets \{v \in V : p(v) = \nullval \}$\;
        $R \gets \{v \in V' : \exists(u,v) \in E(C)\text{ and } v > u.\nextval\}$ and $E_{C,R} \gets E(C, R)$ \;
        $L \gets \{v \in V' : \exists(u,v) \in E(C) \text{ and } v < u.\nextval\}$ and $E_{C,L}\gets E(C, L)$\;
        \smallskip

        $G_R \gets (C, E_{C, R}) \cup \delta_{G[V']}(R)$ and $G_L \gets (C, E_{C, L}) \cup \delta_{G[V']}(L)$\;\label{subgraphs def}
        \smallskip

        \MultiSourceShortestPath{$G_{R}, C$}
        \tcp*{\normalfont{sets $p(v)$ for nodes $v$ in $G_R$ s.t.\ $v\notin C$}} \label{line: shortest path}
        \MultiSourceLongestPath{$G_{L}, C$}
        \tcp*{\normalfont{sets $p(v)$ for nodes $v$ in $G_L$ s.t.\ $v\notin C$}}
        \smallskip

        \lForEach{$v \in C$}{$p(v) \gets u$ where $u \in C$ s.t. $u.\nextval = v$} \label{line: cycle p}
        \lForEach{node $u$ in $G_{L}$ and $G_{R}$}{$u.color \gets$ black} \label{line: color black}
    }

    \smallskip
    
    \Fn{\DFS{u}}{
        $u.color \gets$ gray\; 

        \ForEach{$v \in N(u)$ in increasing order w.r.t.\ $\leq$} {        
            \lIf{$v.color =$ white}{$u.next \gets v$, \DFS{v}} \label{line: dfs call}
            \ElseIf{$v.color =$ gray}{
                $u.next \gets v$, $\bar{u} = u$, and $C \gets \emptyset$\;
                \textbf{do}$\,$ $C$.append($u$), $u \gets u.next$$\,$  \textbf{while}$\,$  $\bar{u} \neq u$ \tcp*{\normalfont{build cycle $C$}}
                % \lDo{$\bar{u} \neq u$}{
                %     $C$.append($u$), $u \gets u.next$\;
                % }
                \ComputeWalks{$C$}
        }}
        $u.color \gets$ black\;\label{DFS sets black}
    }
    $p(u)\gets$ null, $u.color \gets$ white, and $u.next \gets$ null for all $u\in V$\tcp*{\normalfont{initialization}}
    \lForEach{$u \in V$ such that $u.color =$ white}{\DFS{u}}\label{line: all dfs}
\end{algorithm}

Let us denote with $V'\coloneqq \{v \in V : p(v) = \nullval \}$ all the nodes for which we have not yet constructed a leftmost walk. The function \texttt{ComputeWalks} computes two (possibly overlapping) subgraphs $G_R$ and $G_L$ of $G[V']$ and constructs walks from $C$ to the nodes in these two graphs. The graphs $G_R$ and $G_L$ are defined as follows.
%deleted refernce here
First we compute two (possibly overlapping) subsets $R$ and $L$ of $N:=N(C)\cap V'$ as follows. We let $R \coloneqq \{v \in V' : \exists(u,v) \in E(C)\text{ and } v > u.\nextval\}$ and $L \coloneqq \{v \in V' : \exists(u,v) \in E(C) \text{ and } v < u.\nextval\}$, as well as $E_{C,R} \coloneqq E(C, R)$ and $E_{C,L}\coloneqq E(C, L)$. In other words, the set $R$ (the set $L$) consists of those neighbors in $V'$ of nodes $u\in C$ that are larger (smaller) than $u.\nextval$ with respect to $\le$. We then define $G_R = (V_R, E_R) := (C, E_{C, R}) \cup \delta_{G[V']}(R)$ and $G_L = (V_L, E_L) := (C, E_{C, L}) \cup \delta_{G[V']}(L)$.
In order to compute the function $p(u)$ for nodes $u \in V_R \cup V_L \setminus C$, we now launch two BFS-like visits in a specific (and essential) order. We first launch a \emph{multi-source shortest path} search in $G_{R}$ and then a \emph{multi-source longest path} search in $G_L$, in both cases from $C$. In both of these subroutines we set $p(v)=u$ whenever $v$ is discovered to be the next node on a shortest (respectively longest) path starting from the cycle. In both of these BFS-like visits, we process a node's neighbors in increasing order with respect to $\le$ (recall that we are searching for leftmost walks), we give more details in the paragraph below. It is essential that we first run the shortest path search on $G_R$ and then the longest path search on $G_L$ as the two graphs possibly overlap and we hence reset the $p$-values for nodes that are in both graphs, prioritizing paths in $G_L$. 
We then compute $p(u)$ also for the nodes $u \in C$ by setting $p(u) = v$ where $v\in C$ is such that $v.\nextval = u$.
Lastly, we set $u.color = $ black for each $u\in V_{L} \cup V_{R}$. Then the function \texttt{ComputeWalks} terminates and the DFS is resumed. We refer the reader to Figure~\ref{fig:6:left_ex} for an example run of the algorithm.

\subparagraph{Details on the BFS searches.}
The multi-source shortest path search is implemented as a classical BFS on $G_R$ with a queue initially containing the nodes $C$. When a node $u$ is dequeued, nodes $v\in N(u) \cap V_{R}$ are processed in increasing order with respect to $\leq$. If $v$ is newly discovered by the BFS, we set $p(v) = u$ and enqueue $v$.

The multi-source longest path search is implemented as follows. It is essential that $G_L$ is acyclic, see Lemma~\ref{obsacyclic}~($i$). Hence, we can simply compute the maximum distance $d(C,u)$ from $C$ to every node $u$ in $G_{L}$ by traversing $G_{L}$ in a topological order. We then perform a BFS-like search on $G_{L}$ with a queue initially containing the nodes $C$. When a node $u$ is dequeued, nodes $v\in N(u) \cap V_{L}$ are processed in increasing order with respect to $\leq$. If $v$ is newly discovered by the BFS and $d(C,v) = d(C,u) + 1$, we set $p(v) = u$ and enqueue $v$. 

\subparagraph{Analysis.}
We start with the following simple observation.
%\carlo{Thus, for any $u \in V$ for which $p(u) \neq \nullval$, the proposition $u.color$ = black holds, consequently a new invocation of \texttt{DFS($u$)} will not initiate}.
\begin{restatable}[]{observation}{obsiii}\label{dummyleftobs}
    \normalfont{($i$)} At termination of \normalfont{\texttt{ComputeWalks($C$)}} it holds that $p(u) \neq \nullval$ for every $u \in V$ that can be reached from $C$. 
    \normalfont{($ii$)} Before and after each run of \normalfont{\texttt{ComputeWalks}}, if $(u_{i})_{i=1}^{l}$ is a walk in $G$ and $p(u_{j}) \neq \nullval$ for some $j \in [l]$, then $p(u_{1}) \neq \nullval$. 
\end{restatable}

\begin{proof}
    ($i$)~Let $V' = \{u \in V : p(u) = \nullval\}$ as in the algorithm. \texttt{ComputeWalks($C$)} inserts every node $u \in V'$ that can be reached from $C$ to the subgraphs $G_{L}$ or $G_{R}$. Then it computes function $p$ for all and only nodes in $G_{L} \cup G_{R}$.
    % \sout{the only place where the color-attribute of nodes is set to black or the $p$-values are modified are lines~\ref{line: shortest path} to~\ref{line: color black}. Hence, the statement trivially holds before the first call to  \texttt{ComputeWalks}. Thereafter, \texttt{ComputeWalks} only sets $p$ values for nodes in $G_L$ and $G_R$ and all these nodes are labeled black in line~\ref{line: color black}.}
    ($ii$)~When the DFS outputs a cycle $C$, all nodes $u$ that are reachable from $C$ and that satisfy $p(u) =$ null are added to one of the subgraphs $G_{L}$ or $G_{R}$. This holds as $R\cup L = N(C)\cap V'$, where $V'$ were all nodes $v$ with $p(v)=$ null. Hence \texttt{ComputeWalks} always calculates $p(u)$ for all and only those nodes in $G_{L}$ and $G_{R}$.
\end{proof}
We proceed with the lemma about the acyclicity of $G_L$.
\begin{restatable}[]{lemma}{lemiii}\label{obsacyclic}
    Let $G_{L} = (V_{L}, E_{L})$ and $C$ be the subgraph of $G = (V,E)$ and the cycle computed during an arbitrary execution of \normalfont{\texttt{ComputeWalks}} in Forward Visit. Then, \normalfont{($i$)}~$G_{L}$ is acyclic and \normalfont{($ii$)} it holds that $v \notin C$ for each $v$ such that there exists $(u,v) \in E_{L}$.
\end{restatable}

\begin{proof}
    Let $L \subseteq V_{L}$ be as in the algorithm. Suppose for the purpose of contradiction that~($i$) or~($ii$) does not hold. In both cases there exists $z \in V_{L}$ such that $z \in C'$ for some cycle $C'$ in $G$ different from $C$. Let $(u,v) \in  E(C,L)$ be the edge such that $v$ can reach $z$ in $G_L$. As $z \in V_L$, we know that $p(z) = \nullval$ and, by Observation~\ref{dummyleftobs}~($ii$), it follows that $p(\bar{z}) = \nullval$ for each $\bar{z} \in C'$, and, for an analogous argument, also the nodes $\bar v$ that can reach $C'$ must satisfy $p(\bar v) \neq \nullval$. Thus, at this point of the algorithm execution, the function \texttt{ComputeWalks} cannot have colored those nodes black that can reach $C'$, as it colors exactly the nodes for which it computes the function $p$ (i.e., the nodes in $G_{L} \cup G_{R}$). The only remaining part in which Forward Visit may color a node $\bar{v}$ black that can reach $C'$ is line~\ref{DFS sets black}, i.e., when \texttt{DFS($\bar{v}$)} has terminated. Notice however that $\bar v$ can trivially reach $C'$ and thus \texttt{DFS($\bar{v}$)} has to find a cycle $C''$ that can reach the cycle $C'$ in $G$ before terminating (possibly $C''=C'$). Thus, by Observation~\ref{dummyleftobs}~($i$), when \texttt{DFS($u$)} is launched, \texttt{DFS($\bar{v}$)} has not terminated yet, otherwise $p(z) \neq \nullval$ would hold. When \texttt{DFS($u$)} was started, the nodes in $N(u)$ were processed in increasing order with respect to\ $\leq$. In addition, as $v < u.\nextval$ by definition of $L$, it follows that $v$ was visited by \texttt{DFS($u$)} before $u.\nextval$. Finally, since $v$ can reach $z$, it follows that \texttt{DFS($z$)} has to terminate before \texttt{DFS($v$)}. As for each node $\bar v$ reaching $C$ the function \texttt{DFS($\bar{v}$)} has not terminated yet, it holds that \texttt{DFS($v$)} has to output a cycle $C''$ (possibly $C'' = C'$) different from $C$ before terminating, contradicting the assumption that $C$ is the output cycle.
\end{proof}

We next observe that Forward Visit indeed computes a complete function $p$ on $V$.
\begin{restatable}[]{observation}{obsv}\label{obs: p complete}
    Upon termination Forward Visit has computed $p(u)$ for every node $u \in V$. 
\end{restatable}

\begin{proof}
    Let $u \in V$. By assumption, every node has an incoming edge and thus there exists a cycle $C_{u}$ in $G$ containing $u$.
    If before or after any run of \texttt{ComputeWalks} there exists $v \in C_{u}$ with $p(v) \neq \nullval$, then Observation~\ref{dummyleftobs}~($ii$) shows that also $p(u) \neq \nullval$.
    Otherwise, before and after every run of \texttt{ComputeWalks} it holds that $p(v) = \nullval$ for all $v \in C_{u}$.
    In this case, \texttt{ComputeWalks} has never colored the nodes in $C_u$ black, as by line~\ref{line: color black}, it colors black all and only those nodes for which it has computed the function $p$.
    By line~\ref{line: all dfs}, the algorithm starts a \texttt{DFS} visit from every node in $V$ (either in line~\ref{line: dfs call} or in line~\ref{line: all dfs}).
    Therefore, let $v$ be the first node of $C_u$ for which a \texttt{DFS} visit is started.
    As for each $z \in C_u$ it holds that $p(v) = \nullval$, before \texttt{DFS($v$)} starts $z.color = $ white for each $z \in C_u$. However, since $v \in C_u$ the \texttt{DFS} visit cannot terminate without finding a cycle $C$ with $v \in C$. Thus, since $v$ can reach $u$, by property~(i) of Observation~\ref{dummyleftobs}, we can conclude that the next run of \texttt{ComputeWalks} calculates $p(u)$.
\end{proof}

We are now ready to prove Theorem~\ref{exleft}.
% \vspace{-5pt}
\begin{proof}[Proof of Theorem~\ref{exleft}]
    By Observation~\ref{obs: p complete}, it holds that $p(v)\neq $ null for all $v\in V$ upon termination of the algorithm. It remains to prove that $P_{u}=(p(u)^{i})_{i\geq0}$ is a leftmost walk to $u$ according to $\leq$ for any node $u \in V$. Let $u_{i+1} \coloneqq p(u)^{i}$ for $i \geq 0$ (i.e., $u_0 = u$) and suppose for the purpose of contradiction that $P_{u}$ is not a leftmost walk to $u$. Then by Definition~\ref{def:leftmost} there exists a walk $P_{u}' = (\bar{u}_{i})_{i\geq 0}$ to $u$ such that $\bar{u}_{j} = u_{j}$ for some $j>1$ and $\bar{u}_{j-1} < u_{j-1}$. Let $j$ be minimal with that property. Furthermore, let $k$ with $0 \leq k < j -1$ be maximal such that $\bar{u}_{k} = u_{k}$ (such $k$ exists as $\bar{u}_{0} = u_{0} = u$). Consider now the execution of \texttt{ComputeWalks} that has computed $p(u)$ and let the cycle $C$ and the subgraphs $G_{L}=(V_L, E_L)$ and $G_{R}=(V_R, E_R)$ be the instances of those objects in that execution of the function. By Observation~\ref{dummyleftobs}~($ii$), at the beginning of this execution $p(v) = $ null for each $v$ in $P_{u}$ and $P_{u}'$. Consider the walk $\hat P_{u_{k}} = (\hat{u}_{i})_{i\geq k}$ to $u_{k}$, where $\hat{u}_{i} = \bar{u}_{i}$ for $i$ with $k \leq i \leq j$ and $\hat{u}_{i} = u_{i}$ for $i \geq j$. We distinguish three cases: (1)~ $u_{k} \in C$, (2)~$u_{k} \in V_{L} \setminus C$, and (3)~$u_{k} \in V_{R} \setminus (C \cup V_L)$.
    
    (1) Let $u_{k} \in C$. Then, $u_i\in C$ for all $i\ge k$  from how the algorithm sets $p$ in line~\ref{line: cycle p}. Hence, $\hat{u}_{i} \in C$ for all $i \geq j$. Now, consider the largest $m$ with $k \leq m < j$ such that $\hat{u}_{m} = \hat{u}_{m'}$ for some $m < m'$. By assumption, this integer $m$ exists since $k < j$ and $\hat{u}_{k} \in C$. Thus, if $m'$ is the smallest integer satisfying this property, the sequence $C'=\hat{u}_{m}, \hat{u}_{m+1}, \ldots, \hat{u}_{m'}$ is a cycle different from $C$ in $G$.     
    Since $\hat{u}_{j-1} < u_{j-1}$, the DFS has to visit $\hat{u}_{j-1}$ before $u_{j-1}$ and as $\hat{u}_{j-1}$ can reach $\hat{u}_{m}$ and consequently the cycle $C'$, it contradicts that the DFS output $C$. 
    
    (2)~Now, let $u_{k} \in V_{L} \setminus C$ and consider the smallest $h$ with $h > k$ such that $u_{h} \in C$ and observe that then also $u_{i}\in C$ for $i\ge h$. Observe that $P_{u_{k}}' = (u_{i})_{i=k}^{h}$ is a longest path from $C$ to $u_{k}$ in $G_{L}$. We distinguish two sub-cases. 
    (a)~Suppose that $h < j$. Then $u_{j-1} \in C$ and $\hat{u}_{j-1} \in L$ as $\hat{u}_{j-1} < u_{j-1}$. This implies $(\hat{u}_{i+1},\hat{u}_{i}) \in E_{L}$ for each $i$ with $k \leq i < j$ by the definition of $E_L$. By Lemma~\ref{obsacyclic}~($i$), the graph $G_{L}$ is acyclic and thus $(\hat{u}_{i})_{i=k}^{j}$ is a path from $C$ to $\hat{u}_{k}$ in $G_{L}$ that is longer than $P_{u_{k}}'$, a contradiction. 
    (b)~Now assume that $h \geq j$. This implies $\hat{u}_{j-1} \in V_{L} \setminus C$ and trivially also $\hat{u}_{i} \in V_{L}$ for $i$ with $k \leq i < j$. Lemma~\ref{obsacyclic}~($ii$) implies $\hat{u}_{i} \in V_{L} \setminus C$ and hence $h$ is also minimal such that $\hat{u}_{h} \in C$. 
    Due to the previous considerations $d(C, \hat{u}_{i}) = d(C, u_{i})$ for each $i$ with $k \leq i < j$, where $d$ is as in the algorithm. 
    As the multi-source longest path search processes the neighbors of $u_j$ in increasing order with respect to $\leq$ and since $\hat{u}_{j-1} < u_{j-1}$, this contradicts $p(u_{k}) = u_{k+1}$. 
    
    (3) Now assume $u_{k} \in V_{R} \setminus (C \cup V_L)$. Note that $\hat{u}_{k'} \notin V_{L} \setminus C$ for all $k' > k$ as the opposite would imply $u_{k} \in V_{L}$. As in (2) let $h$ be minimal with $h > k$ such that $u_{h} \in C$ and observe that then also $u_{i} \in C$ for each $i \geq h$. Observe that $P_{u_{k}}' = (u_{i})_{i=k}^{h}$ is a shortest path from $C$ to $u_{k}$ in $G_{R}$. We again distinguish two sub-cases. 
    (a)~Suppose that $h < j$. Then $u_{j-1} \in C$ and $\hat{u}_{j-1} \in L$ as $\hat{u}_{j-1} < u_{j-1}$. This implies $\hat{u}_{j-1}  \in V_{L} \setminus C$ as $p(u_{j-1}) = u_{j}$, a contradiction to our observation above. 
    (b)~Now assume that $h \geq j$. 
    This implies $\hat{u}_{j-1} \in V_{R} \setminus C$ and $\hat u_i\in V_R$ for each $i$ with $k \leq i < j$. Furthermore, $\hat u_i\notin C$ for each $i$ with $k \leq i < j$ as $(u_i)_{i = k}^h$ is a shortest path from $C$ to $u_k$. 
    It follows that $(\hat{u}_{i})_{i=k}^{h}$ is another shortest path from $C$ to $u_{k}$ in $G_{R}$. 
    As the multi-source shortest path search processes the neighbors of $u_j$ in increasing order with respect to $\leq$ and since $\hat{u}_{j-1} < u_{j-1}$, this contradicts $p(u_{k}) = u_{k+1}$. 
\end{proof}

%\bibliography{./references}

\end{document}